\documentclass[11pt]{article}

\usepackage{amssymb}
\usepackage{amsmath}
\usepackage{amsthm}
\usepackage{textcomp}
\usepackage{array}
\usepackage[latin5]{inputenc}
\usepackage{lineno}
\usepackage{color}
\usepackage[margin=1in]{geometry}
\usepackage{comment}

\usepackage{graphicx}

\newcommand{\ket}[1]{|#1\rangle}
\newcommand{\braket}[2]{\langle #1|#2\rangle}
\newcommand{\cent}[0]{\mbox{\textcent}}
\newcommand{\dollar}[0]{\$}
\newcommand{\blank}{\#}

\newtheorem{lemma}{Lemma}
\newtheorem{theorem}{Theorem}
\newtheorem{corollary}{Corollary}
\newtheorem{fact}{Fact}

\theoremstyle{definition}

\newcommand{\weak}{\mathsf{weak}\mbox{-}}

\newcommand{\IP}[1]{\mathsf{IP(#1)}}
\newcommand{\perIP}[1]{\mathsf{IP_{1}(#1)}}
\newcommand{\perunIP}[1]{\mathsf{IP_{1,<1}(#1)}}
\newcommand{\IPw}[1]{\weak\mathsf{IP(#1)}}
\newcommand{\perIPw}[1]{\weak\mathsf{IP_{1}(#1)}}

\newcommand{\AM}[1]{\mathsf{AM(#1)}}
\newcommand{\perAM}[1]{\mathsf{AM_{1}(#1)}}

\newcommand{\AMw}[1]{\weak\mathsf{AM(#1)}}

\newcommand{\qAM}[1]{\mathsf{qAM(#1)}}
\newcommand{\perqAM}[1]{\mathsf{qAM_{1}(#1)}}

\newcommand{\perqAMw}[1]{\weak\mathsf{qAM_{1}(#1)}}

\newcommand{\DTIME}[1]{\mathsf{DTIME( #1 )}}

\newcommand{\ATIME}[1]{\mathsf{ATIME( #1 )}}

\newcommand{\DSPACE}[1]{\mathsf{DSPACE( #1 )}}
\newcommand{\NSPACE}[1]{\mathsf{NSPACE( #1 )}}
\newcommand{\ASPACE}[1]{\mathsf{ASPACE( #1 )}}

\newcommand{\qASPACE}[1]{\mathsf{qASPACE( #1 )}}

\newcommand{\ptime}{\mathsf{P}}

\newcommand{\logspace}{\mathsf{L}}
\newcommand{\alogspace}{\mathsf{AL}}
\newcommand{\qalogspace}{\mathsf{qAL}}

\newcommand{\pspace}{\mathsf{PSPACE}}
\newcommand{\apspace}{\mathsf{APSPACE}}
\newcommand{\qapspace}{\mathsf{qAPSPACE}}
\newcommand{\exptime}{\mathsf{EXPTIME}}
\newcommand{\expspace}{\mathsf{EXPSPACE}}

\newcommand{\Next}{\mathtt{next}}

\newcommand{\subsetsum}{\mathtt{SUBSET\mbox{-}SUM}}

\title{Turing-equivalent automata using a fixed-size quantum memory\thanks{A preliminary report on some contents of this paper was \cite{YS11C}.}}
\author{Abuzer Yakary{\i}lmaz\thanks{The author was partially supported by FP7 FET-Open project QCS.}\\
\small University of Latvia, Faculty of Computing, Raina bulv. 19, R\={\i}ga, LV-1586, Latvia 
\\
\small \texttt{abuzer@lu.lv}
\\ \\
\today
}

\date{\small Keywords: 
quantum complexity theory, 
Turing-equivalence,
space-bounded computation,
interactive proof systems, 
Arthur-Merlin games, 
alternation, 
complexity classes, 
finite state automata}

\begin{document}

\maketitle

\thispagestyle{empty}

\begin{abstract}

In this paper, we introduce a new public quantum interactive proof system and the first quantum alternating Turing machine: qAM proof system and qATM, respectively. Both are obtained from their classical counterparts (Arthur-Merlin proof system and alternating Turing machine, respectively,) by augmenting them with a fixed-size quantum register. We focus on space-bounded computation, and obtain the following surprising results: Both of them with constant-space are Turing-equivalent. More specifically, we show that for any Turing-recognizable language, there exists a constant-space weak-qAM system, (the nonmembers do not need to be rejected with high probability), and we show that any Turing-recognizable language can be recognized by a constant-space qATM even with one-way input head.
		
For strong proof systems, where the nonmembers must be rejected with high probability, we show that the known space-bounded classical private protocols can also be simulated by our public qAM system with the same space bound. Besides, we introduce a strong version of qATM: The qATM that must halt in every computation path. Then, we show that strong qATMs (similar to private ATMs) can simulate deterministic space with exponentially less space. This leads to shifting the deterministic space hierarchy exactly by one-level. The method behind the main results is a new public protocol cleverly using its fixed-size quantum register. Interestingly, the quantum part of this public protocol cannot be simulated by any space-bounded classical protocol in some cases.

\end{abstract}

\newpage

\section{Introduction}
\label{sec:intro}

\pagenumbering{arabic}

Anne Condon, in her famous PhD thesis \cite{Co89}, introduced a general computational model,
i.e. \textit{probabilistic game automaton},
that unifies many important computational models and concepts:
\textit{Alternation} of Chandra, Kozen, and Stockmeyer \cite{CKS81}, 
\textit{private alternation} of Reif \cite{Re84},
\textit{Arthur-Merlin games} of Babai \cite{Ba85},
\textit{interactive proof systems} of Goldwasser, Micali, and Rackoff \cite{GMR89},
\textit{game against nature} of Papadimitriou \cite{Pa85}, etc.
In this framework, Arthur-Merlin (AM) proof systems and alternation are the ``weakest",
since both are the games 
with \textit{complete information}.
In this paper, we introduce two new games by augmenting these two models with a fixed-size\footnote{The size of the register does not depend on the length of the input.} quantum register,
namely \textit{qAM} and \textit{q-alternation}, respectively.
We focus our attention to space-bounded computation, and obtain the following surprising results:
Both new games with constant space are \textit{Turing-equivalent}.

Interactive proof (IP) systems and AM proof systems
were introduced by  Goldwasser, Micali, and Rackoff \cite{GMR85}
and Babai \cite{Ba85}, respectively.
In time-bounded computation, it was shown that the class of languages having a polynomial-time IP or AM system
is identical to $ \pspace $ \cite{Sha92}.
In space-bounded computation, IP systems are more powerful than AM systems
for any space-bound \cite{Co89,Co91,DS92}, e.g.
the class of languages having a logarithmic-space AM system is identical to $ \ptime $, and
the class of languages having a logarithmic-space IP system is a superset of $ \exptime $.
It was also shown that \cite{CL89}
for any Turing-recognizable language, there exits a constant-space \textit{weak}\footnote{
	The verifier does not need to halt with high probability for the nonmembers of the corresponding language.
}-IP system.

There are many different definitions of quantum interactive proof (QIP) systems \cite{Kni96,Kit99,Wat99B,AN02,MW05}.
In time-bounded computation, similar to the classical case,
the class of languages having a polynomial-time QIP system were shown to be identical to $ \pspace $ \cite{JJUW11}.
In space-bounded computation, 
the only published work belongs to Nishimura and Yamakami \cite{NY09}.
Their results, unfortunately, are model-dependent, and so do not reflect the full power of QIP systems
since they use some restricted quantum automaton models as the verifiers.

Our qAM proof system is the first public space-bounded QIP system,
and we present the first non-trivial results on space-bounded QIP systems.
We show how a fixed-size quantum register 
leads to unexpected increase in the computational power of a public proof system.
Our main result on the qAM system is that
there exists a constant-space weak-qAM protocol for any Turing-recognizable language.
In the classical case, a similar result is known for constant-space weak private protocols \cite{CL89}.
However, our protocol is not only public but also has perfect-completeness.
Our second result is that for any known $ s(n) $ space-bounded private protocol, 
there exists an equivalent $ s(n) $ space-bounded qAM protocol,
where $ s(n) \in O(1) \cup \Omega(\log(n)) $ is space-constructible.
Therefore, we can say that logarithmic-space is sufficient for qAM systems
for any language in $ \exptime $.

Alternation was introduced independently by Chandra and Stockmeyer \cite{CS76} and Kozen \cite{Ko76}
as a generalization of nondeterminism. It was shown that 
alternation shifts the deterministic hierarchy 
\[
	\logspace \subseteq \ptime \subseteq \pspace \subseteq \exptime \subseteq \expspace
\]
by exactly one level \cite{CKS81}. 
On the other hand, the class of languages recognized by 
alternating finite automata is still the regular languages \cite{CKS81}.
Reif \cite{Re84} introduced private alternation by assuming that
universal player can hide some information from the existential player,
and showed that private alternation shifts the deterministic space hierarchy
\[
	\logspace \subsetneq \pspace \subsetneq \expspace
\]
by exactly one level.

Our q-alternation is \textit{the first definition} of alternation in the domain of quantum computation.
Our main result on q-alternation is that
one-way q-alternating finite automata can recognize any Turing-recognizable language.
In the classical case, the class of languages recognized by any space-bounded (private) ATMs
is a proper subset of decidable languages \cite{Re84}.
Since q-alternating machines may not halt the computation in every path, 
we also introduce the strong version of q-alternation by forbidding infinite computations.
Then, we show that strong q-alternation, similar to private alternation,
shifts the deterministic space hierarchy by exactly one level.

\section{Preliminaries}
\label{sec:pre}

For any string $ x $, $ |x| $ is the length of $ x $ and 
$ x[j] $ is its $ j^{th} $ symbol, where $ 1 \leq j \leq |x| $.
``$ \# $" is the blank symbol.
Moreover, we represent $ O(1) $ with $ \mathsf{1} $ and $ O(\log(n)) $ with $ \mathsf{log} $.

We assume that the reader is familiar with deterministic, nondeterministic, and 
alternating Turing machines, (DTM, NTM, and ATM, respectively,)\footnote{
	We refer the reader to \cite{Re84} for the details of private ATMs although it is not necessary 
	to follow the content.	
} and
their time- and space-bounded complexity classes
$ \mathcal{X}\mathsf{TIME} $ and $ \mathcal{X}\mathsf{SPACE} $, where $ \mathcal{X} $ is
``$ \mathsf{D} $", ``$ \mathsf{N} $", and ``$ \mathsf{A} $", respectively; and,
the following standard classes:
\begin{itemize}
	\item $ \ptime = \cup_{k > 0} \DTIME{n^{k}} $
		and
		$ \exptime = \cup_{k > 0} \DTIME{2^{O(n^{k})}} $;
	\item $ \logspace =  \DSPACE{log} $,
		$ \pspace = \cup_{k > 0} \DSPACE{n^{k}} $, and
		$ \expspace = \cup_{k > 0} \DSPACE{2^{O(n^{k})}} $; 
	\item $ \alogspace =  \ASPACE{log} $ and $ \apspace = \cup_{k > 0} \ASPACE{n^{k}} $.
\end{itemize}

In the following part, we provide the necessary background, based on \cite{DS92,Co93A}, 
for the proof systems.
For a detailed survey on space-bounded interactive proof systems,
we refer the reader to \cite{Co93A}.
An interactive proof system (IPS)  consists of a prover ($ P $) and a verifier ($ V $).
The verifier is a (resource-bounded) probabilistic Turing machine having a read-only input tape, 
a read/write work tape, and a source of random bits.
Each head has two-way access to its tape.
The states of the verifier are partitioned into reading, communication, and halting (accepting or rejecting) states,
and it has a special communication 
cell for communicating
with the prover,
where the capacity of the cell is finite.

The one-step transitions of the verifier can be described as follows.
When in a reading state, the verifier firstly flips an unbiased coin and then
determines its next configuration based on the symbol under the tape heads, 
the state, and the outcome of the coin flip.
When in a communication symbol, 
The verifiers writes a symbol on the communication cell with respect to the current state.
Then, in response, the prover writes a symbol in the cell.
Based on the state and the symbol written by prover, the verifier defines the next state of the verifier.

The prover $ P $ is specified by a prover transition function, 
which determines the response of the prover to the verifier based on the input and 
the verifier's communication history until then.
Note that this function does not need to be \textit{computable}.

For a given input $ x $, the probability that $ (P,V) $ accepts (rejects) $ x $ is the
cumulative accepting (rejecting) probabilities taken over all branches of the verifier.
The prover-verifier pair $ (P,V) $ is an IPS for $ L $
with error probability $ \epsilon < \frac{1}{2} $ if
\begin{enumerate}
	\item[1.] for all $ x \in L $, the probability that $ (P,V) $ accepts $ w $ is greater than $ 1-\epsilon $,
	\item[2.] for all $ x \notin L $, and all provers $ P^* $, the probability that $ (P^*,V) $ rejects $ x $
		is greater than $ 1-\epsilon $.
\end{enumerate}
These conditions are known as completeness and soundness, respectively.
Now, we define some variants of IP systems by restricting and/or relaxing the above conditions.

The prover-verifier pair $ (P,V) $ is an \textit{weak-IPS for $ L $ with error probability}
$ \epsilon < \frac{1}{2} $ if we relax the soundness condition (2) as follows:
\begin{enumerate}
	\item[2$ ' $.] for all $ x \notin L $, and all provers $ P^* $, the probability that $ (P^*,V) $ accepts $ x $
		is at most $ \epsilon $.
\end{enumerate}

The prover-verifier pair $ (P,V) $ is a \textit{(weak or not) IPS having perfect completeness} for $ L $
if we restrict the completeness condition (1) as follows:
\begin{enumerate}
	\item[1$ ' $.] for all $ x \in L $, the probability that $ (P,V) $ accepts $ x $ is exactly equal to 1.
\end{enumerate}

An Arthur-Merlin (AM) proof system (or a public-coin IPS) is a special case of IPS such that after each coin toss, 
the outcome is automatically written on the communication cell, and so the prover can have
complete information about the computation of the verifier.

We will use $ \IP{\cdot} $ and $ \AM{\cdot} $ to represent the space-bounded complexity classes
for IP and AM systems, respectively.
Note that the space bound is always defined on the verifiers.
The ones having perfect-completeness will be shown by 
$ \perIP{\cdot} $ and $ \perAM{\cdot} $, respectively.
We will use prefix ``$ \weak $" to represent their ``weak" versions.

Some known facts, related to our results, on AM and IP systems as given below.
(We also refer the reader to Appendix \ref{app:review-private-protocols} for the details of some private protocols.)
\begin{fact}
	\label{fact:AM}
	\cite{CL88,Co89,DS92}
	For any space-constructible $ s(n) = \Omega(\log n) $, 
	\[
		\AMw{s(n)} = \AM{s(n)} = \ASPACE{s(n)}.
	\]
	Moreover, $ \AMw{1} \subsetneq \AMw{log} = \AM{log} = \ptime $.	
\end{fact}
\begin{fact}
	\label{fact:weakIPconst}
	\cite{CL89}
	$ \mbox{Any Turing recognizable language is in } \IPw{1} $.
\end{fact}
\begin{fact}
	\label{fact:IPconst}
	\cite{DS92} 	
	$ \DTIME{2^{O(n)}} = \ASPACE{O(n)} \subseteq \perIP{1} $.
\end{fact}
\begin{fact}
	\label{fact:perfectIP}
	\cite{Co89,CL89,DS92} 	
	For any space-constructible $ s(n) = \Omega(\log(n)) $,
	\[
		\DTIME{2^{2^{O(s(n))}}} = \ASPACE{2^{O(s(n))}} \subseteq \perIP{s(n)}.
	\]
\end{fact}

\begin{fact}
	\label{fact:IPupperbound}
	\cite{CL89} 
	For any space-constructible $ s(n) = \Omega(\log(n)) $,
	\[
		\perIP{s(n)} \subseteq \IP{s(n)} \subseteq \ATIME{2^{2^{2^{O(s(n))}}}}.
	\]
\end{fact}

As seen from the above facts, private protocols are more powerful than public ones under the same space bounds.
In case of weak-soundness,
the power of the private protocols with finite-state verifiers becomes Turing-equivalent,
which is never possible for a public protocol with any given space bound.

We will shortly show that the public protocols using a fixed-size quantum register
can also implement some private protocols under the same space bounds.
More specifically, the results presented in Facts \ref{fact:IPconst} and \ref{fact:perfectIP} can also be obtained
for qAM systems. 
Moreover, in case of weak-soundness,
our new public protocol can also be Turing-equivalent even restricting to perfect-completeness,
which can never be a case for private protocols with any given space bound 
due to Theorem \ref{thm:char-of-perfectIP} (see below).

\begin{theorem}
	\label{thm:char-of-perfectIP}
	For any space-constructible $ s(n) \in \Omega(\log(n)) $,
	\[
		\perIP{s(n)} \subseteq \perIPw{s(n)} \subseteq \ASPACE{2^{2^{O(s(n))}}}.
	\]
\end{theorem}
\begin{proof}
	See Appendix \ref{app:proof-of-char-of-perfectIP}.
\end{proof}

Note that Theorem \ref{thm:char-of-perfectIP} improves the previously known upper bound (Fact \ref{fact:IPupperbound})
for space-bounded IPS with perfect-completeness.

\section{qAM}
\label{sec:qAM}

In this section, we give the definition of our new AM system (qAM) 
and present our results on qAM proof systems.

A qAM (\textit{q}uantum Arthur-Merlin)\footnote{The small ``q" indicates that the verifier has 
a ``very small" (possible the smallest) quantum resource.} 
proof system is an AM system where the verifier additionally has a fixed size quantum register. 
The reading state of the new verifier is as follows:
\begin{center}
	\begin{minipage}{0.9\textwidth}
		A superoperator, determined by the current state and the symbol(s) under the tape head(s),
		is applied to the quantum register, 
		and the outcome of the operator is automatically written on the communication cell
		in order to satisfy \textit{the complete information requirement}.
		Then, the next configuration is determined based on 
		the the current state, the symbol(s) under the tape head(s),
		and the observed outcome.
		For any deterministic transition, 
		the verifier applies an identity operator on the register.
		We refer the reader to Figure \ref{fig:superoperators} for the details of superoperators.
	\end{minipage}
\end{center}
Note that the verifier no longer needs a classical random source.\footnote{ 
	For example, the superoperator 
	$
		\mathcal{E} = \left\lbrace E_{h_{1}} = \frac{1}{2} I, E_{h_{2}} = \frac{1}{2} I,
		E_{t_{1}} = \frac{1}{2} I, E_{t_{2}} = \frac{1}{2}I \right\rbrace
	$
	always produces the outcomes \textit{head} (``$ h_1 $" or ``$ h_2 $)
	and \textit{tail} (``$ t_1 $" or ``$ t_2 $") with probability $ \frac{1}{2} $.
}
We will use $ \qAM{\cdot} $ and $ \perqAM{\cdot} $ to represent 
qAM counterparts of $ \AM{\cdot} $ and $ \perAM{\cdot} $, respectively.
Prefix ``$ \weak $" is also applicable to $ \qAM{\cdot} $ and $ \perqAM{\cdot} $.
We give a simple qAM protocol for 
the well-known NP-complete language $ \subsetsum $ in Appendix \ref{app:qAM-for-subsetsum}.

\begin{figure}[!ht]
	\centering
	\footnotesize
	\fbox{
	\begin{minipage}{0.97\textwidth}
		The most general quantum operator is a superoperator,
		which generalizes stochastic and unitary operators and also includes measurement.
		Formally, a superoperator $ \mathcal{E} $ is composed by a finite number of operation elements,
		$ \mathcal{E} = \{ E_{1}, \ldots, E_{k} \} $, satisfying that
		\begin{equation}
		\label{eq:completeness}
			\sum_{i=1}^{k} E_{i}^{\dagger} E_{i} = I,
		\end{equation}
		where $ k \in \mathbb{Z}^{+} $ and the indices are the measurement outcomes.
		When a superoperator, say $\mathcal{E}$, is applied to 
		the quantum register in state $\ket{\psi} $, i.e. $ \mathcal{E}(\ket{\psi}) $,
		we obtain the measurement outcome $i$ with probability 
		$ p_{i} = \braket{\widetilde{\psi_{i}}}{\widetilde{\psi_{i}}} $,
		where $\ket{\widetilde{\psi_{i}}}$, \textit{the unconditional state vector}, 
		is calculated as $ \ket{\widetilde{\psi}_{i}} = E_{i} \ket{\psi} $ and $1 \leq i \leq k$.
		(Note that using unconditional state vector simplifies calculations in many cases.)
		If the outcome $i$ is observed ($p_{i} > 0 $), the new state of the system 
		is obtained by normalizing $ \ket{\widetilde{\psi}_{i}} $, 
		which is $ \ket{\psi_{i}} = \frac{\ket{\widetilde{\psi_{i}}}}{\sqrt{p_{i}}} $.
		Moreover, as a special operator, the quantum register can be initialized to a predefined quantum state.
		This initialize operator, which has only one outcome, is denoted $ \acute{\mathcal{E}} $.
		In this paper, the entries of quantum operators are defined by rational numbers.
		Thus the probabilities of the outcomes are always rational numbers.
	\end{minipage}
	}
	\caption{The details of superoperators}
	\label{fig:superoperators}
\end{figure}

Knowledgeable readers will 
have noticed that, when the verifier is restricted to use constant
space, the qAM system  is actually
the quantum counterpart of the finite automaton with both
nondeterministic and probabilistic states
of Condon et. al. \cite{CHPW98}, and that if we remove the communication
with the prover as well we end up with a finite automaton
with quantum and classical states (2QCFA) of Ambainis and Watrous \cite{AW02}.

In our qAM protocols (and later in our q-alternation simulations),
we use some non-unitary transformations to implement our main tasks.
We define our superoperators based on these transformations.
Let $ E_1, \ldots, E_k $ be some of these transformations.
We can obtain a superoperator $ \mathcal{E} $ based on them 
by defining some additional transformations $ E_{k+1}, \ldots, E_{k+k'} $
such that 
\[
	\mathcal{E} = \left\lbrace 
		\frac{1}{d} E_1, \ldots, \frac{1}{d} E_k, \frac{1}{d} E_{k+1}, \ldots, \frac{1}{d} E_{k+k'}  
		\right\rbrace
\]
satisfies Condition \ref{eq:completeness} in Figure \ref{fig:superoperators} 
for a convenient $ d > 1 $, where $ k' > 0 $.\footnote{
	We refer the reader \cite{YS10A,YS11A} for similar procedures.
}
We call $ \left\lbrace \frac{1}{d} E_1, \ldots, \frac{1}{d} E_k \right\rbrace $
\textit{the main operation elements}
and $ \left\lbrace \frac{1}{d} E_{k+1}, \ldots, \frac{1}{d} E_{k+k'} \right\rbrace $
\textit{the auxiliary operation elements}.
Moreover, in our protocols,
the computation continues on the quantum register only when 
the outcomes of some main operation elements are observed.
On the other hand, the current computation on the quantum register is always terminated/restarted with discarding
the current content of the register when the outcome of an auxiliary operation element is observed.
Therefore, the details of the auxiliary operation elements can be omitted from the description of the protocols.

In the following part, we present our qAM protocols.
We begin with a constant-space weak-qAM protocol having perfect completeness
for any given Turing-recognizable language.
We present our result by giving \textit{a new public protocol} simulating a given DTM.
Contrary to the classical case,\footnote{
	In classical case, to show universality of constant-space weak-IP systems,
	a simulation of two-way finite automaton with two-counters was given \cite{CL89}.
	Since the complexity classes are defined by Turing machines,
	this technique does not seem to be applicable to the case of strong-soundness.
	(See also Appendix \ref{app:review-private-protocols})
} our simulation technique can also be applicable to the case of strong-soundness.
Therefore, after making certain modifications, we present our other space-bounded qAM protocols.
Note that, all of our qAM protocols have \textit{perfect-completeness}.\footnote{
	Two-sided bounded error is necessary for the universality of constant-space weak-IP systems
	due to Theorem \ref{thm:char-of-perfectIP}.
}
Interestingly, this property allows us to use the same techniques for q-alternation 
after making some restrictions and modifications.

\begin{theorem}
	\label{thm:weak-qAM}
	Any Turing recognizable language is in $ \perqAMw{1}  $.
\end{theorem}
\begin{proof}
	Let $ \mathtt{L} $ be a Turing recognizable language and $ \mathcal{D} $ be a single-tape DTM
	recognizing $ \mathtt{L} $.
	We will construct a weak-qAM proof system $ (P,V) $  for $ \mathtt{L} $
	with perfect completeness, where $ V $ is a finite state verifier.
	
	We begin with some details of $ \mathcal{D} $.
	$ Q $ containing $ q_{1} $ (the initial state), $ q_{a} $ (the accepting state), 
	and $ q_{r} $ (the rejecting state) is the set of states, 
	$ \Gamma $ containing $ \blank $ is the tape alphabet, and
	$ \Gamma' = Q \cup \Gamma $, called configuration alphabet, 
	where $ Q $ and $ \Gamma $ are disjoint sets and $ \dollar \notin \Gamma $.
	Note that $ \Gamma' $ contains at least 5 elements.
	Any configuration of $ \mathcal{D} $ is of the form
	$ uqv $ ($ \mathcal{D} $ is in $ q $ and the tape head is on the leftmost symbol of $ v $),
	where $ q \in Q $ and $ uv \in \blank (\Gamma)^{*} \blank $.
	The unnecessary blank symbols are always dropped from the descriptions of configurations. 
	For a given input string $ x $,
	the initial configuration is represented as $ q_{1} \blank x \blank $.
	
	The main protocol is executed in an infinite loop and each iteration (round) is composed by the following:
	(i) The verifier requests the computation (a sequence of configurations starting 
	from the initial configuration) of $ \mathcal{D} $ on the given input, say $ x $, from the prover.
	(ii) Against the cheating provers, the verifier checks the correctness of the computation and rejects 
	$ x $ if it detects a defect in the computation.
	(iii) When it encounters a halting configuration, the verifier mimics the decision of this (halting) configuration.
	
	Let $ w $ be the string obtained from the prover in a single round. 
	The verifier expects $ w $ as $ c_1 \dollar \dollar c_2 \dollar \dollar c_3 \dollar \dollar \cdots $, 
	where 
	(P1) $ c_i $'s ($ i>0 $) are some configurations of $ \mathcal{D} $,
	(P2) $ c_1 $ is the initial configuration, and
	(P3) $ c_{i+1} $ is 
	the successor of $ c_i $ 
	in one step for any $ i>0 $.
	(We use double $ \dollar $ for pedagogical reasons.)
	Note that $ w $ can be an infinite string.
	The verifier can check P1 and P2 deterministically, and 
	$ x $ is rejected immediately if one of them fails.
	Therefore, in the following part, we assume that $ w $ satisfies both P1 and P2 and
	each configuration ends with ``$ \dollar \dollar $".
	
	The non-trivial part is to check P3 for each $ i>0 $, 
	i.e. whether $ c_{i+1} $ is identical to $ \Next(c_{i}) $,
	where $ \Next(c_i) $ 
	is the single-step successor of $ c_i $.
	This is where the quantum register comes into play.
	The idea behind is to encode $ \Next(c_i) $ and $ c_{i+1} $ into the amplitudes of two states on the register, 
	and then to subtract them, and 
	to reject $ x $ with the resulting amplitude.\footnote{In fact,
	the encoding part can also be implemented by a probabilistic system but 
	the aforementioned \textit{subtraction} is not possible in classical systems!}
	We call this procedure \textit{successor-check}. 
	The $ i^{th} $ successor-check compares $ \Next(c_i) $ and $ c_{i+1} $.
	As will be detailed below, 
	whenever $ c_{i+1} = \Next(c_i) $ (for all defined $ i>0 $), 
	the input is never rejected by a successor-check.
	Thus, once a halting configuration is obtained from the prover, the \textit{only decision} is made based on it.
	Otherwise, i.e. $ \exists i>0 $ s.t. $ c_{i+1} \neq \Next(c_i) $,
	$ x $ is rejected by the $ i^{th} $ successor-check.
	We show that such a reject probability is sufficiently greater than 
	any accept probability in a single round. 
	
	In the remaining part, we will give the details of a single round and the analyses of the protocol.
	The verifier requests the computation of $ \mathcal{D} $ on $ x $ symbol by symbol from the prover, and
	it uses a 3-symbol buffer to parallelly encode the legal successor of
	the presently scanned configuration.	
	The verifier uses a 4-state quantum register, $ q_1,\ldots,q_4 $,
	which is set to $ \ket{ \psi_{1,0} } = (1 ~~ 0 ~~ 0 ~~ 0)^{T} $ at the beginning of each round.	
	The configurations are encoded in base-$ m $, where $ m = |\Gamma'| +1 $.
	Each symbol of $ \Gamma' $ is associated with a different positive integer.
	Thus, any configuration $ c_{i} $ ($ i>0 $) can be represented by a $ |c_i| $-length number in base-$ m $.
	We use the same symbol for both the encoded string/symbol and its encoding.
	The verifier applies one superoperator per symbol.
	When the outcome of an auxiliary operation elements is observed,
	the verifier terminates the current round and initiates a new round.
	The tasks implemented by the main operation elements reduce the amplitudes with $ \frac{1}{d} < 1 $.
	Therefore, after applying each superoperator, a new round is initiated with some probability.
	In other words, a round can continue only with a small probability.
	The complete details of the superoperators and the related operations are given at the end of the proof
	due to their technicalities.
	
	Let $ l_{i} $ be the length of 
	$ c_1 \dollar \dollar c_2  \dollar \dollar \cdots  c_{i} \dollar \dollar $ ($ i>0 $).
	By processing  $ c_1 \dollar \dollar $, 
	i.e.  a series of superoperators $ \mathcal{E}_{1,1}, \ldots,\mathcal{E}_{1,|c_1 \dollar \dollar|} $
	are applied to the register, 
	$ \Next(c_1) $ is encoded into the amplitudes of $ \ket{q_2} $.
	Then, the quantum state becomes\footnote{
		Note that, unconditional state vectors facilitate the calculations.
		The probabilities can be calculated directly.
	} 	
	\[
	 	\ket{ \widetilde{ \psi_{2,0} } } = \left( \frac{1}{d} \right)^{l_1} 
		\left( \begin{array}{c} 1 \\ \Next(c_1) \\ 0 \\ 0 \end{array}		 \right).
	\]
	Similarly, by processing $ c_2 \dollar $, $ c_2 $ and $ \Next(c_2) $ are encoded into the amplitudes of 
	$ \ket{q_3} $ and $ \ket{q_4} $, respectively:
	
	\[
		\ket{ \widetilde{ \psi_{2,|c_2 \dollar|} } } =
		\left( \frac{1}{d} \right)^{l_2-1} 
		\left( \begin{array}{c} 1 \\ \Next(c_1) \\ c_2 \\ \Next(c_2) \end{array}		 \right).
	\]
	After processing one more $ \dollar $, the first successor-check is finalized:
	The corresponding superoperator has two main operation elements.
	The first one is responsible for comparing $ \Next(c_1) $ and $ c_2 $,
	and subtracts the amplitudes of $ \ket{q_2} $ and $ \ket{q_3} $.
	Its outcome is observed with probability
	\[
		\left( \frac{1}{d} \right)^{2 l_2} (\Next(c_1)-c_2)^{2},
	\]
	which is also the rejecting probability of the first successor-check.
	The second main operation element is determined conditionally.
	If $ \Next(c_2) $ is an accepting (a rejecting) configuration,
	then it selects only the amplitude of $ \ket{q_1} $,
	the outcome of which is observed with probability $ \left( \frac{1}{d} \right)^{2 l_2} $.
	This probability is also the accepting (rejecting) probability of the round.
	Since the computation is terminated due to the outcomes of both operation elements,
	the round does not continue in this case.
	If $ \Next(c_2) $ is not a halting configuration,
	the round continues with the next successor-check.
	Thus the quantum state becomes
	\[
		\ket{ \widetilde{ \psi_{3,0} } } =
		 \left( \frac{1}{d} \right)^{l_2} 
		\left( \begin{array}{c} 1 \\ \Next(c_2) \\ 0 \\ 0 \end{array}		 \right).
	\]
	One can easily verify that if the prover is cheating about $ c_2 $, the input is rejected 
	with a probability at least $ m^{2} \left( \frac{1}{d} \right)^{2 l_2} $ 
	since both $ \Next(c_1) $ and $ c_2 $ end with a $ \# $ and
	they must be disagree on at least one digit.	
	If the prover is honest, the decision is given only if $ \Next(c_2) $ is a halting configuration,
	whose probability is at least $ m^{2} $ smaller than the above rejecting probability.
	The other successor-checks are executed exactly in the same way (see the end of the proof).
	
	The analysis of the protocol is as follows.
	If $ x \in L $, then $ P $ always sends the correct computation of $ \mathcal{D} $ on $ x $ to $ V $,
	say $ c_1 \dollar\dollar c_2 \dollar\dollar \cdots \dollar\dollar c_{t} \dollar\dollar $
	such that $ \Next(c_t) $ is an accepting configuration.
	Then $ V $ never rejects but accepts $ x $ with probability $ \left( \frac{1}{d} \right)^{2l_{t}} $ in each round.
	So, $ x $ is accepted exactly by $ V $.
	
	If $ x \notin L $, then the only case in which $ V $ accepts $ x $ is that $ P^{*} $
	send a computation of some nonhalting configurations 
	$ c_1 \dollar\dollar c_2 \dollar\dollar \cdots \dollar \dollar c_{t'} \dollar\dollar $
	such that $ \Next(c_{t'}) $ is an accepting configuration, where $ t'>1 $.
	Since $ x \notin L $, there must be an $ i ~ (1 \leq i < t') $ such that 
	$ \Next(c_i) \neq c_{i+1} $. 
	Therefore, $ x $ is rejected with probability at least $ m^{2} $ times greater than the accepting probability
	in a single round.
	In other words, the overall accepting probability can be bounded above by $ \frac{1}{m^{2}+1} $.
	This bound can be easily reduced to any desired value by selecting a greater $ m $ value.
	
	Now, we give the omitted details of the superoperators and the related operations below.
	Remember that	
	$ l_i $ is the length of $ c_1 \dollar \dollar c_2 \dollar \dollar \cdots c_i \dollar \dollar  $
	and the quantum state is set to 
	\[
		\ket{\psi_{1,0}} = \left( \begin{array}{c} 1 \\ 0 \\ 0 \\ 0 \end{array}	 \right)
	\]
	at the beginning of each round.
	As described before, we omit the details of each auxiliary operation element,
	and a new round is initiated when the outcome of such an operation element is observed.
	
	For each symbol of $ w_{1} = c_1 \dollar \dollar $,
	the verifier applies $ | w_{1} | $ superoperators,
	$ \mathcal{E}_{1,1},\ldots,\mathcal{E}_{1,|w_1|} $, to the register, some of the the same.
	The aim is to encode $ \Next(c_1) $ into the amplitudes of $ \ket{q_2} $.
	For each $ j \in \{ 1, \ldots, |c_1|-1 \} $, the main operation element of 
	$ \mathcal{E}_{1,j} $ is as follows:
	\[
		\frac{1}{d} \left( 
			\begin{array}{cccr}
				1 & ~~0~~ & ~~0~~ & ~~0 \\ 
				\Next(c_{1})[j] & m & 0 & 0 \\ 
				0 & 0 & 0 & 0 \\ 
				0 & 0 & 0 & 0
			\end{array} 
		\right).
	\]
	For $ \mathcal{E}_{1,|c_1|} $ and $ \mathcal{E}_{1,|c_1 \dollar|} $, we have the following cases,
	each of which can be deterministically determined and handled by using 3-symbol buffer:
	\begin{itemize}
		\item If $ |\Next(c_1)| = |c_1|-1 $, the main operation elements of 
		$ \mathcal{E}_{1,|c_1|} $ and $ \mathcal{E}_{1,|c_1 \dollar|} $ are as follows:
			\[
				\frac{1}{d} \left( 
				\begin{array}{cccr}
					1 & ~~0~~ & ~~0~~ & ~~0 \\ 
					0 & 1 & 0 & 0 \\ 
					0 & 0 & 0 & 0 \\ 
					0 & 0 & 0 & 0
				\end{array} 
				\right) \mbox{ and }
				\frac{1}{d} \left( 
				\begin{array}{cccr}
					1 & ~~0~~ & ~~0~~ & ~~0 \\ 
					0 & 1 & 0 & 0 \\ 
					0 & 0 & 0 & 0 \\ 
					0 & 0 & 0 & 0
				\end{array} 
				\right),
			\]
			respectively,
			since the encoding of $ \Next(c_1) $ is finished by superoperator 
			$ \mathcal{E}_{1,|c_1|-1} $. 
		\item  If $ |\Next(c_1)| = |c_1| $, the main operation elements of 
			$ \mathcal{E}_{1,|c_1|} $ and $ \mathcal{E}_{1,|c_1 \dollar|} $ are as follows:
			\[
				\frac{1}{d} \left( 
					\begin{array}{cccr}
						1 & ~~0~~ & ~~0~~ & ~~0 \\ 
						\# & m & 0 & 0 \\ 
						0 & 0 & 0 & 0 \\ 
						0 & 0 & 0 & 0
					\end{array} 
				\right)
				\mbox{ and }
				\frac{1}{d} \left( 
				\begin{array}{cccr}
					1 & ~~0~~ & ~~0~~ & ~~0 \\ 
					0 & 1 & 0 & 0 \\ 
					0 & 0 & 0 & 0 \\ 
					0 & 0 & 0 & 0
				\end{array} 
				\right),
			\]
			respectively,
			since the encoding of $ \Next(c_1) $ is finished by superoperator 
			$ \mathcal{E}_{1,|c_1|} $. 
			(Note that the last symbol of any configuration is a blank symbol.)
		\item If $ |\Next(c_1)| = |c_1|+1 $, the main operation elements of 
			$ \mathcal{E}_{1,|c_1|} $ and $ \mathcal{E}_{1,|c_1 \dollar|} $ are as follows:
			\[
				\frac{1}{d} \left( 
					\begin{array}{cccr}
						1 & ~~0~~ & ~~0~~ & ~~0 \\ 
						\Next(c_1)[|c_1|] & m & 0 & 0 \\ 
						0 & 0 & 0 & 0 \\ 
						0 & 0 & 0 & 0
					\end{array} 
				\right)
				\mbox{ and }
				\frac{1}{d} \left( 
				\begin{array}{cccr}
					1 & ~~0~~ & ~~0~~ & ~~0 \\ 
					\# & m & 0 & 0 \\ 
					0 & 0 & 0 & 0 \\ 
					0 & 0 & 0 & 0
				\end{array} 
				\right),
			\]
			respectively,
			since the encoding of $ \Next(c_1) $ is finished by  superoperator 
			$ \mathcal{E}_{1,|c_1|+1} $. 
	\end{itemize}
	The main operation elements of $ \mathcal{E}_{1,|c_1 \dollar \dollar|} $ is as follows:
	\[
		\frac{1}{d} \left( 
			\begin{array}{cccr}
				1 & ~~0~~ & ~~0~~ & ~~0 \\ 
				0 & 1 & 0 & 0 \\ 
				0 & 0 & 0 & 0 \\ 
				0 & 0 & 0 & 0
			\end{array} 
		\right).
	\]
	Thus, after applying superoperators $ \mathcal{E}_{1,j} $'s ($ 1 \leq j \leq |w_{1}| $),
	the state vector of the register becomes
	\begin{equation}
		\label{eq:state-vector-before-2}
		\ket{ \widetilde{ \psi_{2,0} } } = \left( \frac{1}{d} \right)^{l_1} 
		\left( \begin{array}{c} 1 \\ \Next(c_1) \\ 0 \\ 0 \end{array}	 \right).
	\end{equation}
	
	We continue with the block $ w_2 = c_2 \dollar \dollar $.
	In fact, the tasks implemented in this part are the same for any other block $ w_i = c_i \dollar \dollar $ ($ i >2 $).
	Similar to the above case,
	for each symbol of $ w_2 $, the verifier applies $ |w_2| $ superoperators,
	$ \mathcal{E}_{2,1}, \ldots, \mathcal{E}_{2,|w_2|} $, to the register, some of which can be the same.
	Remember that before starting to apply any operator,
	the unconditional state vector is $ \ket{ \widetilde{ \psi_{2,0} } } $ (Equation \ref{eq:state-vector-before-2}).
	The aims in this block ($ w_2 $) are as follows:
	
	\begin{enumerate}
		\item By processing $ c_2 \dollar $,
		\begin{enumerate}
			\item to encode $ c_2 $ into the amplitudes of $ \ket{q_3} $ and
			\item to encode $ \Next(c_2) $ into the amplitudes of $ \ket{q_4} $.
		\end{enumerate}
		\item By processing the second $ \dollar $,
		\begin{enumerate}
			\item to finalize the $ 1^{st} $ successor-check,
			\item to accept or reject the input if $ \Next(c_2) $ is an accepting or a rejecting configuration, 
				respectively; and,
				to start $ 3^{rd} $ successor-check if $ \Next(c_2) $ is not a halting configuration.
		\end{enumerate}	
	\end{enumerate}
	The details of superoperators to encode $ c_2 $ and $ \Next(c_2) $ are similar to the ones given above.
	For each $ j \in \{ 1, \ldots, |c_2|-1 \} $, the main operation element of 
	$ \mathcal{E}_{2,j} $ is as follows:
	\[
		\frac{1}{d} \left( 
			\begin{array}{cccr}
				1 & ~~0~~ & ~~0~~ & ~~0 \\ 
				0 & 1 & 0 & 0 \\ 
				c_2[j] & 0 & m & 0 \\ 
				\Next(c_2)[j] & 0 & 0 & m
			\end{array} 
		\right).
	\]
	For $ \mathcal{E}_{2,|c_2|} $ and $ \mathcal{E}_{2,|c_2 \dollar|} $, we have the following cases:
	\begin{itemize}
		\item If $ |\Next(c_2)| = |c_2|-1 $, the main operation elements of 
		$ \mathcal{E}_{2,|c_2|} $ and $ \mathcal{E}_{2,|c_2 \dollar|} $ are as follows:
			\[
				\frac{1}{d} \left( 
				\begin{array}{cccr}
					1 & ~~0~~ & ~~0~~ & ~~0 \\ 
					0 & 1 & 0 & 0 \\ 
					\# & 0 & m & 0 \\ 
					0 & 0 & 0 & 1
				\end{array} 
				\right) \mbox{ and }
				\frac{1}{d} \left( 
				\begin{array}{cccr}
					1 & ~~0~~ & ~~0~~ & ~~0 \\ 
					0 & 1 & 0 & 0 \\ 
					0 & 0 & 1 & 0 \\ 
					0 & 0 & 0 & 1
				\end{array} 
				\right),
			\]
			respectively,
			since the encoding of $ \Next(c_2) $ is finished by the superoperator $ \mathcal{E}_{2,|c_2|-1} $.
		\item  If $ |\Next(c_2)| = |c_2| $, the main operation elements of 
			$ \mathcal{E}_{2,|c_2|} $ and $ \mathcal{E}_{2,|c_2 \dollar|} $ are as follows:
			\[
				\frac{1}{d} \left( 
					\begin{array}{cccr}
						1 & ~~0~~ & ~~0~~ & ~~0 \\ 
						0 & 1 & 0 & 0 \\ 
						\# & 0 & m & 0 \\ 
						\# & 0 & 0 & m
					\end{array} 
				\right)
				\mbox{ and }
				\frac{1}{d} \left( 
				\begin{array}{cccr}
					1 & ~~0~~ & ~~0~~ & ~~0 \\ 
					0 & 1 & 0 & 0 \\ 
					0 & 0 & 1 & 0 \\ 
					0 & 0 & 0 & 1
				\end{array} 
				\right),
			\]
			respectively,
			since the encoding of $ \Next(c_2) $ is finished by the superoperator $ \mathcal{E}_{2,|c_2|} $. 
		\item If $ |\Next(c_2)| = |c_2|+1 $, the main operation elements of 
			$ \mathcal{E}_{2,|c_2|} $ and $ \mathcal{E}_{2,|c_2 \dollar|} $ are as follows:
			\[
				\frac{1}{d} \left( 
					\begin{array}{cccr}
						1 & ~~0~~ & ~~0~~ & ~~0 \\ 
						0 & 1 & 0 & 0 \\ 
						\# & 0 & m & 0 \\ 
						\Next(c_2)[|c_2|] & 0 & 0 & m
					\end{array} 
				\right)
				\mbox{ and }
				\frac{1}{d} \left( 
				\begin{array}{cccr}
					1 & ~~0~~ & ~~0~~ & ~~0 \\ 
					0 & 1 & 0 & 0 \\ 
					0 & 0 & 1 & 0 \\ 
					\# & 0 & 0 & m
				\end{array} 
				\right),
			\]
			respectively,
			since the encoding of $ \Next(c_2) $ is finished by superoperator $ \mathcal{E}_{2,|c_2|+1} $. 
	\end{itemize}
	Thus, before applying $ \mathcal{E}_{2,|c_2\dollar\dollar|} $, the state vector becomes as follows:
	\begin{equation}
		\label{eq:state-vector-end-2}
		\ket{ \widetilde{ \psi_{2,|c_2 \dollar|} } } = \left( \frac{1}{d} \right)^{l_2-1} 
		\left( \begin{array}{c} 1 \\ \Next(c_1) \\ c_2 \\ \Next(c_2) \end{array}	 \right).
	\end{equation}
	Operator $ \mathcal{E}_{2,|c_2\dollar\dollar|} $ has two main operation elements.
	This first one is responsible to finalize the $ 1^{th} $ successor-check:
	\[
		\frac{1}{d} \left( 
					\begin{array}{ccrr}
						0 & ~~0~~ & 0 & ~~0 \\ 
						0 & 1 & -1 & 0 \\ 
						0 & 0 & 0 & 0 \\ 
						0 & 0 & 0 & 0
					\end{array} 
				\right).
	\]
	The associated action of this operation element is \textit{to reject the input}.
	Therefore, the input is rejected with probability 
	\[
		\left( \frac{1}{d} \right)^{2l_2} \left( \Next(c_1) - c_2 \right)^{2},
	\]
	which is zero if the check succeeds ($ \Next(c_1) = c_2 $) and
	is at least 
	\[
		\left( \frac{1}{d} \right)^{2l_2} m^{2}
	\]
	if the check fails ($ \Next(c_1) \neq c_2 $).
	Since the last symbol of $ \Next(c_1) $ and $ c_2 $ are the identical,
	the value of $ | \Next(c_1) - c_2 | $ can be at least $ m $.
	
	The second main operation element is determined by the type of $ \Next(c_2) $:
	\begin{itemize}
		\item If it is an accepting (a rejecting) configuration, then the following operation element is applied:
			\[
				\frac{1}{d} \left( 
					\begin{array}{ccrr}
						1 & ~~0~~ & 0 & ~~0 \\ 
						0 & 0 & 0 & 0 \\ 
						0 & 0 & 0 & 0 \\ 
						0 & 0 & 0 & 0
					\end{array} 
				\right).
			\]
			The associated action of this  is \textit{to accept (reject) the input}.
			Therefore, the input is accepted (rejected) with probability 
			\[
				\left( \frac{1}{d} \right)^{2l_2}.
			\]
			Note that the round is certainly terminated in this case.
		\item If it is not a halting configuration, then the following operation element is applied:
			\[
				\frac{1}{d} \left( 
					\begin{array}{ccrr}
						1 & ~~0~~ & 0 & ~~0 \\ 
						0 & 0 & 0 & 1 \\ 
						0 & 0 & 0 & 0 \\ 
						0 & 0 & 0 & 0
					\end{array} 
				\right).
			\]
			The associated action of this is to continue the round, and so \textit{to initiate} the $ 3^{rd} $ successor-check.
			The state vector becomes
			\[
				\ket{ \widetilde{ \psi_{3,0} } } = \left( \frac{1}{d} \right)^{l_2} 
				\left( \begin{array}{c} 1 \\ \Next(c_2) \\ 0 \\ 0 \end{array}	 \right).
			\]
	\end{itemize}
	Note that if the prover never sends $ \dollar\dollar $ symbols, then no decision is given.
	
	The tasks for block $ w_3 = c_3 \dollar \dollar $ is exactly the same as for block $ w_2 = c_2 \dollar \dollar $,
	and the tasks for block $ w_4 = c_4 \dollar \dollar $ is exactly the same as for block $ w_3 = c_3 \dollar \dollar $,
	and so on.
	Therefore we can generalize it for a generic block $ w_i = c_i \dollar \dollar $, where $ i \geq 2 $.
	The state vector is 
	\[
		\ket{ \widetilde{ \psi_{i,0} } } = \left( \frac{1}{d} \right)^{l_{i-1}} 
				\left( \begin{array}{c} 1 \\ \Next(c_{i-1}) \\ 0 \\ 0 \end{array}	 \right)
	\]
	at the beginning.
	The aims are as follows:
	\begin{enumerate}
		\item By processing $ c_i \dollar $,
		\begin{enumerate}
			\item to encode $ c_i $ into the amplitudes of $ \ket{q_3} $ and
			\item to encode $ \Next(c_i) $ into the amplitudes of $ \ket{q_4} $.
		\end{enumerate}
		\item By processing the second $ \dollar $,
		\begin{enumerate}
			\item to finalize the $ (i-1)^{st} $ successor-check,
			\item to accept or reject the input if $ \Next(c_{i}) $ is an accepting or a rejecting configuration, 
				respectively; and, 
				to start the $ (i+1)^{st} $ successor-check if $ \Next(c_i) $ is not a halting configuration.
		\end{enumerate}	
	\end{enumerate}
	When the $ (i-1)^{st} $ successor-check is finalized,
	the input is rejected with probability 
	\[
		\left( \frac{1}{d} \right)^{2l_i} ( \Next(c_{i-1}) - c_i )^{2},
	\]
	which can be at least
	\[
		\left( \frac{1}{d} \right)^{2l_i} m^{2}
	\]
	if $  \Next(c_{i-1}) \neq c_i  $.
	If $ \Next(c_{i}) $ is an accepting (a rejecting) configuration,
	then the input is accepted (rejected) with probability
	\[
		\left( \frac{1}{d} \right)^{2l_i}.
	\]
	Otherwise, the $ (i+1)^{st} $ successor-check initialized with the state vector
	\[
		\ket{ \widetilde{ \psi_{i,0} } } = \left( \frac{1}{d} \right)^{l_{i}} 
				\left( \begin{array}{c} 1 \\ \Next(c_{i}) \\ 0 \\ 0 \end{array}	 \right).
	\]
\end{proof}

One of the remarkable properties of our protocol is that 
the quantum register can be in a \textit{superposition}  of two successor-checks,
which is one of the fundamental and distinctive properties of quantum computation.
In fact, classical private protocols can also ``imitate" this phenomenon:
The verifier privately selects odd- or even-numbered successor-checks, and so,
\textit{from the viewpoint of the prover}, the verifier seems to be in a superposition of two successor-checks
(such as, with probability $ \frac{1}{2} $).
However, in our protocol, the verifier \textit{really} is in a superposition and it is independent from any prover. 
This is indeed why our protocol works in a public setting.
Therefore, our protocol can be seen as a new and elegant evidence on 
how the superposition phenomenon can become useful 
in terms of complexity theory.

The main consequence of Theorem \ref{thm:weak-qAM} is that qAM systems can be more powerful than IP systems:
\begin{corollary}
	For any space bound $ s(n) $,
	$ \perIPw{s(n)} \subsetneq \perqAMw{1}. $
\end{corollary}

Now, we turn our attention to the protocols in which the computation always halts with high probability.
(Note that, our weak-protocol may run forever in some cases, i.e.,
the simulated machine may run forever on the given input, 
a cheating prover never sends $ \dollar \dollar $ after sending a few valid configurations, etc.)
In our weak-protocol (above),
the input head of the verifier is never used after checking the first configuration.
In fact, it can be used to check whether 
the length of a configuration sent by the prover is linear.
Thus, the verifier can force the prover to send linear-size configurations.
So, it can be easily obtained that $ \DSPACE{n} \subseteq \perqAM{1} $,
i.e. the prover sends the computation of a linear-space DTM.
If the prover is additionally asked to provide nondeterministic choices,
this result is extended to that $ \NSPACE{n} \subseteq \perqAM{1} $,
i.e. the prover sends the computation of a linear-space NTM,
in which nondeterministic choices are determined by the prover.
Although it is not trivial as the above two results,
we can go one further step:
$ \ASPACE{n} \subseteq \perqAM{1} $,
i.e. the prover sends the computation of a linear-space ATM,
in which nondeterministic choices are determined by the prover
and universal choices are determined by the verifier.
This idea was firstly given by Reif \cite{Re84} for the simulation of a space-bounded ATM by
a private ATM, and then used by Condon and Ladner \cite{CL88}, Condon \cite{Co89}, 
and Dwork and Stockmeyer \cite{DS92} for similar simulations.
We follow the latest result by embedding the simulation idea of Dwork and Stockmeyer \cite{DS92}
into our weak-protocol by making some modifications.

\begin{theorem}
	\label{thm:strong-qAM-const}
	$ \DTIME{2^{O(n)}} = \ASPACE{n} \subseteq \perqAM{1} $.
\end{theorem}
\begin{proof}
	Let $ \mathtt{L} $ be a language in $ \ASPACE{n} $.
	Then, there exists a single-tape ATM $ \mathcal{A} $ recognizing $ \mathtt{L} $.
	The components of $ \mathcal{A} $ is similar to $ \mathcal{D} $ 
	given in the proof of Theorem \ref{thm:weak-qAM}.
	(Note that, for an ATM,
	any nonhalting state is labelled as either existential or universal.)
	We make some additional assumptions on $ \mathcal{A} $
	as given in Dwork and Stockmeyer \cite{DS92}
	(See also Appendix \ref{app:review-private-protocols}).
	Tape alphabet $ \Gamma $, apart from $ \# $, contains another special symbol $ \cent $.
	The input, say $ x $, is given as $ \cent x \cent $, 
	$ \cent $ is only overwritten by $ \cent $,
	the head is not allowed to leave the area between the cent symbols,
	and
	any non-cent symbol is only overwritten by a non-cent symbol.	
	Thus, any configuration of $ \mathcal{A} $ is of the form $ uqv $
	such that $ q \in Q $ and $ uv \in \cent ( \Gamma \setminus \{\cent\} )^{|x|} \cent $.
	We assume that $ \mathcal{A} $ makes an existential move after a universal one and vice versa.
	Moreover, each such a transition leads to exactly two branches.
	In order to fix the running time in each branch, 
	we assume that $ \mathcal{A} $ has some additional deterministic states 
	(i.e. universal states with no branching)
	to keep a counter to guarantee that
	the decision is always given after making $ 2^{c|x|} $ branching transitions.
	So, we can group non-halting states of $ \mathcal{A} $ into three groups:
	\begin{enumerate}
		\item existential states,
		\item universal states leading to two transitions, and 
		\item universal states leading to one transition.
	\end{enumerate}
	The third ones are called deterministic states to prevent confusion.
	So, the second ones are called just universal states.
	We also assume that each counter operation takes the same amount of steps, which is only dependent on 
	the length of input ($ |x| $).
	So the computation tree of $ \mathcal{A'} $ on $ x $, denoted by $ \mathcal{T}_{\mathcal{A}}(x) $,
	has $ 2^{2^{c|x|}} $ leafs,
	and each path from the root to a leaf has the same depth,
	where $ c $ is an appropriate constant.
	If the leaf is an accepting (rejecting) configuration, then the path is called accepting (rejecting) path.
	
	Now we will describe a qAM proof system $ (P,V) $ for $ \mathtt{L} $
	with perfect completeness based the qAM system given in the proof of Theorem \ref{thm:weak-qAM}
	after making some modifications, where $ V $ is a finite state verifier.
	
	The first modification is that
	the verifier requests a computation path of $ \mathcal{T}_{\mathcal{A}}(x) $,
	i.e. a path from the root to one of the leafs, from the prover in each round.
	The format of the computation is the same:
	\[
		c_1 \dollar \dollar c_2 \dollar \dollar \cdots c_i \dollar \dollar c_{i+1} \dollar \dollar \cdots .
	\]
	However, before a successor-check, say the $ i^{th} $ one,
	the verifier should know which branch it follows after $ c_i $
	since the verifier encodes $ \Next(c_i) $ during processing $ c_i \dollar \dollar $.
	Therefore, the verifier should know the follow-up branch before obtaining $ c_i \dollar \dollar $.
	Similarly, the prover should be aware of this branch before sending $ c_{i+1} $.
	As mentioned earlier, any existential branch is determined by the prover and 
	any universal branch is determined by the verifier.
	One of the convenient way is that
	before communicating on $ c_i \dollar \dollar $,
	both parties exchange their decisions \textit{about which branch is followed after $ c_i $}.
	Remark that if $ c_i $ is a deterministic branch, then both choices become useless,
	if it is an existential (a universal) one, the choice of the prover (the verifier) becomes useless.
	The verifier makes its choice with equal probability.
	Therefore, it applies a superoperator having the following two main operation elements
	\[
		\left\lbrace  \underbrace{ \frac{1}{d} I }_{l} , \underbrace{ \frac{1}{d} I }_{r} \right\rbrace,
	\]
	where outcome ``l" (``r") represents the left-branch (right-branch) and
	$ I $ is the identity operator.
	
	The second modification is that 
	each configuration length is \textit{deterministically} checked 
	by the verifier to be equal to $ |x|+2 $ by help of the input head.
	(We denote this property as P4.)
	If not, the computation is terminated and the input is rejected immediately.
	
	The third modification is about the acceptance probability,
	which is dramatically made smaller when compared to the one in the original protocol.
	We describe why the original strategy does not work with an example.
	\begin{center}
	\begin{minipage}{0.9\textwidth}
		Remember that 
		if a round ends with an accepting (a rejecting) configuration in the original protocol,
		then the input is accepted (rejected) with a probability
		calculated by the amplitude of $ \ket{q_1} $.
		Suppose that the prover sends a computation of $ \mathcal{A} $ satisfying P1-P4.
		Then if the verifier follows the original strategy,
		the input is rejected (accepted) with the same probability at the end of each round
		since the length of each path is equal in this case.
		If the prover follows a nondeterministic strategy of $ \mathcal{A} $ that leads to
		exactly one rejecting leaf, then the input is accepted with a high probability
		although the input must be rejected with high probability with respect to this subtree.
	\end{minipage}
	\end{center}
	Our new acceptance strategy is as follows.
	The verifier uses an additional state, $ q_5 $, on the register.
	When a new round is initiated,
	the state vector is immediately set to
	\[
		 \left( \frac{1}{d} \right) \left( \begin{array}{c} 1 \\ 0 \\ 0 \\ 0 \\ 1 \end{array}	 \right).
	\]
	Then, the amplitude of $ \ket{q_5} $ is multiplied by
	\[
		\left( \frac{1}{2d} \right)
	\]
	for each configuration in the computation.
	The input is accepted similar to the original protocol but by the amplitudes of $ \ket{q_5} $.
	Note that after making $ b $ branches, 
	the ratio of the amplitude of $ \ket{q_1} $ to the amplitude of $ \ket{q_5} $
	is $ 2^{b} $, and so,
	any rejecting probability calculated by the amplitude of $ \ket{q_1} $
	is $ 4^{b} $ times greater than 
	any accepting probability calculated by the amplitude of $ \ket{q_1} $.

	Now, we can analyse our modified protocol.
	If $ x \in L $, then (honest) $ P $ always sends a valid computation of $ \mathcal{A} $ on $ x $,
	and the input is always accepted with a (constant) non-zero probability in each round.
	Therefore, $ x $ is accepted by $ V $ exactly.
	
	If $ x \notin L $, there exists always at least 
	one rejecting path whichever nondeterministic strategy is followed.
	As known form the original protocol,
	if the prover ($ P^{*} $) sends some configurations violating any of P1-P4, 
	then the rejecting probability is always sufficiently greater than the accepting one in a single round.
	Therefore, suppose that the prover sends the configurations satisfying P1-P4.
	Let $ p $ be the probability of rejecting $ x $ at the end of a rejecting path.
	So, the accepting probability of $ x $ at the end of an accepting path
	is equal to 
	\[
		p(\frac{1}{4})^{2^{c|x|}}.
	\]
	Since there can be at most $ 2^{2^{c|x|}}-1 $ accepting paths, a single rejecting probability 
	is sufficiently greater than the overall accepting probability.
\end{proof}

\begin{theorem}
	\label{thm:strong-qAM-s}
	For any space-constructible $ s(n) \in \Omega(\log(n)) $,
	\[ \DTIME{2^{2^{O(s(n))}}} = \ASPACE{2^{O(s(n))}} \subseteq \perqAM{s(n)}. \]
\end{theorem}
\begin{proof}
	In this case, the verifier can force the prover to send $ 2^{O(s(n))} $-size configurations
	by using its classical work tape. 	
	The remainder follows from the proof of Theorem \ref{thm:strong-qAM-const}.
\end{proof}
\begin{corollary}
	$ \mathsf{EXPTIME} \subseteq \perqAM{log} $.
\end{corollary}

Due to Fact \ref{fact:AM} and Theorem \ref{thm:strong-qAM-s}, we can follow that
qAM systems (having perfect-completeness) are \textit{strictly} more powerful than any AM system under the same
space bound.
\begin{corollary}
	For any space bound $ s(n) $,
	\[
		\AM{s(n)} \subseteq \AMw{s(n)}  \subsetneq  \perqAM{s(n)}.
	\]
\end{corollary}

\section{q-Alternation}
\label{sec:q-alternation}

As mentioned earlier, alternation and AM proof systems are games with \textit{complete information}.
Moreover, they could be obviously related to each other, e.g.
alternation can be ``inherited" from AM systems as follows:
The verifier is replaced by a universal player and
all provers are represented by an existential player.\footnote{We refer the reader to Condon \cite{Co89}
for the technical details.}

In this section, 
we introduce the notion of \textit{quantum alternation for the first time}.
We define quantum alternation similar to our qAM system,
i.e., its quantum part is only a fixed-size quantum register and 
this register can only be accessible by the universal states.
We call the model \textit{q-alternation} due to its ``very" limited quantum part,
and give the definition based on Turing machine, \textit{q-alternating Turing Machine} (qATM).

A qATM is an ATM augmented with a fixed-size quantum register, 
based on which universal branches are determined.
Any configuration of a qATM can be represented by a pair $ (c,\ket{\psi}) $,
where $ c $ represent the classical configuration of the machine and 
$ \ket{\psi} $ is the state of quantum register.
Let $ \{c_1,\ldots,c_{k_c}\} $ be the transitions determined by the classical transition function of the machine
with respect to the classical state and the symbol(s) under the tape head(s) in configuration $ c $,
where $ k_c $ is the total number branches.
If $ c $ is an existential configuration, then
only the classical part of the qATM is changed,
and the following transition(s) is (are) implemented:
\[
	(c,\ket{\psi}) \rightarrow \{ (c_i,\ket{\psi}) \mid 1 \leq i \leq k_c\}.
\]
\normalsize
If $ c $ is a universal configuration,
then both the classical part and the quantum state of the qATM are changed:
The machine applies a superoperator determined by 
the classical state and the symbol(s) under the tape head(s) in configuration $ c $, 
$ \mathcal{E}_c = \{ E_{c,1},\ldots,E_{c,k_c} \} $,
to the register, i.e. 
\[
	\ket{\psi_i} = \frac{\ket{\widetilde{\psi_i}}}{\sqrt{p_i}}
	\mbox{ if } 
	p_i \neq 0,
	\mbox{ where }
	p_i = \braket{ \widetilde{\psi_i} }{ \widetilde{\psi_i} },
	~
	 \ket{ \widetilde{ \psi_i } } = E_{c,i} \ket{\psi},
 	\mbox{ and }
 	1 \leq i \leq k_c,
\]
\normalsize
and then the following transition(s) is (are) implemented:
\[
	(c,\ket{\psi}) \rightarrow \{ (c_i,\ket{\psi_i}) \mid p_i > 0, 1 \leq i \leq k_c \}.
\]
\normalsize
Note that, the transitions having zero probability $ (p_i = 0) $ are not implemented.\footnote{
	In terms of two-person games \cite{Co89},
	we can say that \textit{the player who makes the universal choices uses a quantum register to make its choices,
	therefore any choice with zero probability can never be a part of that player's strategy}.
}
The computation starts when the machine is in the initial classical configuration and 
the initial quantum state. 
The computation is terminated with the decision of ``acceptance" (``rejection")
if the machine enters an accepting (rejecting) configuration.
The acceptance criteria of qATM is the same as ATM. 
For any given input, we have a computation tree representing all moves of the machine.
The input is accepted if and only if 
there exists a \textit{finite accepting subtree}\footnote{
	Each leaf of an accepting subtree is an accepting leaf, a leaf in which the decision of ``acceptance" is given.
} for a nondeterministic strategy in this computation tree.
If we remove the work tape of a qATM, and restrict the input head to one-way, 
we obtain a one-way q-alternating finite automaton (q-1AFA).

We begin with a q-1AFA simulation of the qAM protocol given in the proof of Theorem \ref{thm:weak-qAM}.
Thus, we obtain that q-alternation leads us to simulate any TM even 
the input head is restricted to one-way.

\begin{theorem}
	\label{thm:RE-q-1AFA}
	Any Turing-recognizable language can be recognized by a q-1AFA.
\end{theorem}
\begin{proof}
	Let $ \mathtt{L} $ be a Turing-recognizable language,
	$ \mathcal{D} $ be a single-tape DTM recognizing $ \mathtt{L} $, and
	$ (P,V) $ be the qAM system for $ L $, as described in the proof of Theorem \ref{thm:weak-qAM}.
	It is obvious that $ V $ never needs to move its input head to the left in a single round.
	That is, the input head is used only at the beginning of each round to check
	whether the prover sends the valid initial configuration, 
	which can be easily be implemented by moving the input head from the left to the right once.
	We define a new one-way finite state verifier $ V' $ based on $ V $.
	The only difference between $ V $ and $ V' $ is that
	when the outcome of an auxiliary operation element is observed,
	$ V' $ terminates the computation with decision of ``acceptance", instead of initiating a new round.
	Thus, $ V' $ executes only a single round 
	(and so $ V' $ never needs to move its input head to the left).
	
	The analysis of the proof system $ (P,V') $ is as follows.
	Let $ x $ be an input string. 
	If $ x \in \mathtt{L} $, the computation is terminated in every branch,
	and the input is accepted by $ V' $ with probability 1 by the help of $ P $.
	(Remember that $ (P,V) $ has perfect-completeness.)
	If $ x \notin \mathtt{L} $, 
	there are two cases depending on the prover ($ P^* $) strategy 
	and also the behaviour of $ \mathcal{D} $ on $ x $:
	(1) the computation may be run forever in some branches, and,
	(2) the computation is terminated in every branch and
	the input is rejected by $ V' $ with some non-zero probability .
	
	Now, based on $ V' $, we can easily construct a q-1AFA $ \mathcal{A} $ recognizing $ \mathtt{L} $.
	The universal states of $ \mathcal{A} $ simulate $ V' $ and
	the existential states of $ \mathcal{A} $ simulate the communications with all possible provers.
	If $ x \in \mathtt{L} $,
	there exists a finite \textit{accepting} subtree 
	whose existential moves correspond to the communication with $ P $.
	If $ x \notin \mathtt{L} $,
	the finite sub-tree(s) can only be possible in the second case given above.
	Obviously, at least one leaf of such a subtree is a reject.
	Therefore, there is no finite accepting subtree for the nonmembers.
\end{proof}

\begin{corollary}
	For any space bound $ s(n) $,
	q-1AFAs are strictly more powerful than any $ s(n) $ space-bounded private ATM.
\end{corollary}
\begin{proof}
	This follows from Theorem \ref{thm:RE-q-1AFA} and the fact that
	the class of languages recognized by any space-bounded private ATM
	is a proper subset of decidable languages \cite{Re84},
\end{proof}
Since the entries of any operation element are rational numbers,
any space-bounded qATM can be simulated by a DTM in a straightforward way.
Due to this fact and Theorem \ref{thm:RE-q-1AFA},
we cannot mention a space hierarchy for q-alternation.
Moreover, the computation of qATMs may not be halted in some paths.
Therefore, we define a restricted version of q-alternation:
\textit{strong q-alternation}.
Any q-alternating machine is a strong one if it halts on every computational path.
We will denote the related space complexity classes by $ \qASPACE{\cdot} $,
i.e. $ \qASPACE{s(n)} $ is the class of languages recognized by $ s(n) $ space-bounded strong qATMs.
$ \qalogspace $ and $ \qapspace $ are strong q-alternating counterparts of 
$ \alogspace $ and $ \apspace $, respectively.
We show that 
strong q-alternation (similar to private alternation \cite{Re84}) shifts the deterministic space hierarchy
by exactly one level.

\begin{theorem}
	\label{thm:qatm:qaspace-equal-exp-dspace}
	For any space-constructible $ s(n) \in \Omega(\log(s(n))) $,
	\begin{center}
$
	 \DSPACE{2^{O(s(n))}} = \qASPACE{s(n)}.
	$
\end{center}
\end{theorem}
\begin{proof}
	From \cite{CKS81} and Lemma \ref{lem:qatm:lower-bound} (see below), we can follow that
	\begin{center}
$
		\DSPACE{2^{O(s(n))}} \subseteq \ATIME{s(n)} \subseteq \qASPACE{s(n)}.
	$
\end{center}
	From Lemma \ref{lem:qatm:q-ATM-simulated-by-NTM} (see below) and Savitch's theorem \cite{Sav70}, we can follow that
	\begin{center}
$
		\qASPACE{s(n)} \subseteq \NSPACE{2^{O(s(n))}} \subseteq \DSPACE{2^{O(s(n))}}.
	$ 
\end{center} 
\end{proof}

\begin{corollary}
	$ \logspace \subsetneq \qalogspace = \pspace \subsetneq \qapspace = \expspace $.
\end{corollary}

In the remaining part, we give some technical lemmata used in the proof of 
Theorem \ref{thm:qatm:qaspace-equal-exp-dspace}.
We begin with showing an upper bound on the running time of a space-bounded strong qATM.
\begin{lemma}
	\label{lem:qatm:upper-bound}
	For any $ s(n) \in \Omega(\log (n)) $,
	the running time of a $ s(n) $ space-bounded strong qATM can be at most $ 2^{O(s(n))} $.
\end{lemma}
\begin{proof}
	The proof follows from Appendix \ref{app:time-bound-for-AH-QTMs}.
\end{proof}

Due to Lemma \ref{lem:qatm:upper-bound},
we can also provide a nondeterministic space simulation of a given space-bounded strong qATM
by exponential blow-up.

\begin{lemma}
	\label{lem:qatm:q-ATM-simulated-by-NTM}
	For any space-constructible $ s(n) $,
	if $ \mathtt{L} $ is recognized by a $ s(n) $ space-bounded strong qATM $ \mathcal{A} $,
	there exists a $ 2^{O(s(n))} $ space-bounded NTM $ \mathcal{N} $ recognizing $ \mathtt{L} $.
\end{lemma}
\begin{proof}	
	Let $ x $ be a given input.
	The length of any computation path of $ \mathcal{A} $ on $ x $ can be at most $ 2^{c_1 s(|x|)} $ 
	and any configuration length can be at most $ c_2s(|x|) $,
	for appropriate numbers $ c_1 $ and $ c_2 $.
	We know that if $ x \in \mathtt{L} $, there exists a nondeterministic strategy
	that leads to an accepting subtree, and,
	if $ x \notin \mathtt{L}  $, the subtree of each nondeterministic strategy
	has at least one rejecting leaf (a leaf in which the decision of ``rejection" is given).
	$ \mathcal{N} $ nondeterministically implements each strategy of $ \mathcal{A} $ on $ x $,
	and, for each strategy, traces the corresponding subtree path-by-path by using $ 2^{O(s(|x|))} $ space.
	Note that $ 2^{O(s(|x|))} $ space is also sufficient to trace the content of the quantum register in a path
	since in each step, the precision of any amplitude can be increased by at most a constant,
	and so the space to hold the state of the register increases at most by a constant in each step.
	
	If $ \mathcal{N} $ detects a rejecting leaf for a strategy, then it rejects the input.
	If there is no such a leaf for a strategy, then it accepts the input.
	Therefore, if $ x \in \mathtt{L} $, 
	$ \mathcal{N} $ accepts the input in at least one of its nondeterministic branch;
	and, if $ x \notin \mathtt{L} $, the input is rejected in all nondeterministic branches.
\end{proof}

Now, we show that the bound given in Lemma \ref{lem:qatm:q-ATM-simulated-by-NTM} is actually tight.

\begin{lemma}
	\label{lem:qatm:lower-bound}
	For any log-space constructible $ t(n) \in \Omega(n) $,
	if $ L $ is recognized by a ATM $ \mathcal{A} $ running in time $ t(n) $,
	then there exists a $ O(\log (t(n))) $ space-bounded strong qATM $ \mathcal{A'} $ recognizing $ L $.
\end{lemma}
\begin{proof}
	We know that any $ s(n) $ space-bounded ATM can be simulated by a 
	$ O(\log (s(n))) $ space-bounded qAM proof system, say $ (P,V) $,
	with perfect-completeness (Theorem \ref{thm:strong-qAM-s}).
	A generic schema of this simulation is as follows:
	\begin{equation*}
		\label{eq:loop}
		\mbox{
			\begin{minipage}{0.9\textwidth}
				\small
				\texttt{BEGIN LOOP} 
				\\
				\hspace*{10pt}
				\texttt{$ V $ obtains a computation path of the ATM on the given input from the prover}
				\\
				\hspace*{10pt}
				\texttt{$ V $ processes this computation and makes a decision with some probability}
				\\
				\hspace*{10pt}
				\texttt{IF $ V $ makes a decision, THEN the computation (LOOP) is terminated}
				\\
				\texttt{END LOOP} 
			\end{minipage}			
		}
	\end{equation*}
	Let $ x $ be an input.
	In this simulation, $ V $ can deterministically check weather the length of
	of a configuration sent by the prover is $ c s(|x|) $ by using its work tape, 
	where $ c $ is an appropriate number.
	On the other hand, the maximum length of a single-round is determined by the prover.
	For example, for a valid computation, a honest prover can send $ 2^{O( s(|x|) )} $ configurations,
	and $ V $ can only count until $ O( s(|x|) ) $ by using a ``standard" counter.
	
	If we replace the simulated ATM with our $ t(n) $ time-bounded ATM $ \mathcal{A} $, 
	then the same protocol can still work with space bound $ O(\log(t(n))  ) $.
	Now, the maximum length of a single-round can be determined by the verifier
	since it can count $ t(|x|) $ in this case and can terminate the computation with decision of ``rejection"
	if the prover does not sent a halting configuration of $ \mathcal{A} $ until then. 
	We denote this new proof system as $ (P',V') $.
	By using the idea given in the proof of Theorem \ref{thm:RE-q-1AFA},
	we can define a new verifier $ V'' $ based on $ V' $ such that
	it terminates the computation with the decision of ``acceptance" when
	it observes the outcome of an auxiliary operation element,
	and so it implements only a single-round of $ (P',V') $.
	
	The analysis of the proof system $ (P',V'') $ is as follows.
	The computation is terminated in every branch.
	If $ x \in \mathtt{L} $, it is accepted by $ V'' $ with probability 1 (due to perfect-completeness)
	by the help of $ P' $.
	If $ x \notin \mathtt{L} $, 
	the input is always rejected by $ V' $ with some non-zero probability.

	As described in the proof of Theorem \ref{thm:RE-q-1AFA},	
	we can easily construct a $ O( \log( t(|x|) )) $ space-bounded qATM $ \mathcal{A'} $  based on $ V'' $, 
	which recognizes $ \mathtt{L} $:
	The universal states of $ \mathcal{A'} $ simulate $ V'' $, and
	the existential states of $ \mathcal{A} $ simulate the communications with all possible provers.
	If $ x \in \mathtt{L} $,
	there exists a finite \textit{accepting} subtree 
	whose existential moves correspond to the communication with $ P' $.
	If $ x \notin \mathtt{L} $,
	any finite subtree contains at least one rejecting leaf.
\end{proof}

\textbf{Acknowledgements.}
We would like to thank Andris Ambainis and A. C. Cem Say for their 
many helpful comments on some drafts of this paper.


\appendix

\section{A review of previous private protocols}
\label{app:review-private-protocols}

In this part, we review the private protocols for obtaining the results given in Facts 
\ref{fact:weakIPconst} and \ref{fact:IPconst}, which we will refer as
\textit{the weak-protocol} and \textit{the strong-protocol}, respectively.
(Since Fact \ref{fact:perfectIP} is a generalization of Fact \ref{fact:IPconst} \cite{DS92}, we omit its details.)
The main idea behind both protocols for a given language is to simulate a fixed machine 
recognizing the given language:
The prover sends to the verifier the computation of the simulated machine on a given input,
which is a sequence of configurations starting from the initial configuration,
and the verifier tries to verify the correctness of this sequence
and then gives a decision with respect to the halting configuration sent by the prover.
During the verification procedure,
the following tasks can easily be checked deterministically:
(i) each configuration has a correct format, and
(ii) the sequence starts with the initial configuration and ends with a halting configuration.
In the latter protocol,
the length of each configuration is also checked to be at most linear.
On the other hand, 
to check whether the prover sends a valid next configuration after each one is a non-trivial task.
We will call this task \textit{successor-check} and this is indeed 
where the private coin-flips come into play.

Since a single computation requires many adjunctive successor-checks and 
each of them contains many (private) coin-flips,
a decision on the input can only be given with a very small probability after passing a single computation.
Therefore, the prover sends the computation repeatedly in an infinite loop.
Depending on the machine and resources of the the verifier, 
either the verifier can always halt the computation with high probability
or the protocol can run forever in some cases.
For example, if the simulated machine never halts on the input and the prover honestly sends the corresponding computation, the verifier can never halt and give a decision.

Now, we give some protocol specific details.
We begin with the weak-protocol of Condon and Lipton \cite{CL89}.
In his seminal paper \cite{Fr81}, Freivalds presented a two-way probabilistic finite automaton (pfa)
recognizing language
\[
	\mathtt{FRE} = \{ a^{n_1}b^{n_1} a^{n_2}b^{n_2} \cdots a^{n_k}b^{n_k} 
		\mid n_1,\ldots,n_k >0, k>0  \}
\]
with bounded error.
Based on Freivalds' algorithm, 
Condon and Lipton proposed a private protocol such that
if a prover sends a member of $ \mathtt{FRE} $ repeatedly to a one-way pfa verifier,
then the verifier detects the memberships of the input with high probability.
If the prover sends some nonmembers of $ \mathtt{FRE} $ repeatedly to the same verifier,
then the verifier gives a decision of rejection with high probability.
Two-way finite automata with two-counters (2D2CA) are Turing-equivalent \cite{Mi67},
and their main configuration elements are the contents of the counters, which can be encoded unary.
Then successor-checks on a computation of a 2D2CA can be implemented by the protocol given by Condon and Lipton.
Let $ \mathtt{L} $ be a Turing-recognizable language and $ \mathcal{D} $ be the 2D2CA recognizing it.
For the members of $ \mathtt{L} $, 
the (honest) verifier sends finite valid configurations, and so the verifier accepts the input 
them with high probability.
For the nonmembers of $ \mathtt{L} $,
the verifier can only accept if the last configuration of the computation is an accepting one.
This means that the computation contains at least one defect, 
and so the probability of rejection is greater than the probability of acceptance due to the defect.
Since a prover can never sends a halting configuration, the protocol has a weak-soundness.

In the case of the strong-protocol of Dwork and Stockmeyer \cite{DS92}, 
the simulated machine is an $ O(n) $ space-bounded ATM, say $ \mathcal{A} $.
In order to simplify the proof, some inessential assumptions are made on $ \mathcal{A} $:
Roughly, each existential or universal transition leads to exactly two branches,
there is no consecutive two existential or universal branching,
$ \mathcal{A} $ uses only $ |x|+2 $ space, 
$ \mathcal{A} $ always halts exactly after $ 2^{c|x|} $ branching steps by keeping a counter, and so
$ \mathcal{A} $ never enters a loop, where $ x $ is a given input and $ c $ is an appropriate constant.
Note that the computation tree of $ \mathcal{A} $ on $ x $ has $ 2^{2^{c|x|}} $ leafs.
The prover repeatedly sends the computation of some paths of this tree, 
in which the nondeterministic choices are made by the prover 
and the universal choices are made by the verifier, i.e. 
the verifier flips a coin and sends the outcome to the verifier.
The configurations on a computation are separated by $ \dollar $'s as
\[
	\cdots \dollar c_{i} \dollar c_{i+1} \dollar c_{i+2} \dollar \cdots,
\]
where $ i>0 $.
Since each configuration of $ \mathcal{A} $ on $ x $ is fixed length ($ |x|+3 $), 
the verifier can implement the successor-check by deterministically comparing the symbols 
at distance $ |x|+4 $ by using its input head.\footnote{If the simulated machine is 
$ s(n) \in \omega(n) $ space-bounded, then
$ \log(s(n)) $ space-bounded verifier is sufficient for a similar check,
where $ s $ is a space-constructible function.}
If there is not exactly one $ \dollar $ between two symbols, 
then the input is rejected by the verifier.
The private part of the successor-check is that 
the verifier always selects the first symbol randomly and
repeats this selection after each comparison.
That is, if $ l>0 $ is the index of a selected symbol,
then it is compared with $ (l+|x|+4)^{th} $ symbol, and then
a new symbol with an index between $ (l+|x|+4) + 1 $ and $ (l+|x|+4) + |x|+4 $ is selected.
Since the protocol is private, the prover can never know which two symbols are compared.
Thus, any defect on the computation can always be detected by the verifier with some probability,
which is sufficiently greater than any accept probability.
If there is no defect, the accepting probability is still very small so that 
a single rejecting probability on a leaf 
can always dominate the cumulative accepting probabilities from all the other leafs.
(We omit the details here
but the same issue is also detailed in the proof of Theorem \ref{thm:strong-qAM-const}.)
Therefore, for the strings accepted by $ \mathcal{A} $,
none of successor-check produces a rejecting probability, and so
the verifier accepts the input exactly by help of a honest prover.
For the strings rejected by $ \mathcal{A} $, 
if the input is not rejected by any successor-check,
the input is rejected at least one path which is sufficient to dominate all the accepting decisions.
Note that, if the prover never sends a halting configuration,
the verifier detects infinitely many defects, which are sufficient to reject the input with probability 1,
by using its input head.

As seen from the details, the private methods used by the protocols are different.
In the former one, the verifier privately collects statistical evidence from over the computations 
to decide their correctness. 
Moreover, two-sided error is necessary in this case due to Theorem \ref{thm:char-of-perfectIP}.
In the latter case, the verifier can force the prover to send linear size configurations and then
detects the defects directly.
The latter protocol has also perfect-completeness.

\section{The proof of Theorem \ref{thm:char-of-perfectIP}}
\label{app:proof-of-char-of-perfectIP}

The proof of 
$
	\perIP{s(n)} \subseteq \perIPw{s(n)} \subseteq \ASPACE{2^{2^{O(s(n))}}}
$
follows from Lemma \ref{lem:app:proof-of-char-of-perfectIP} (see below)
and \cite{CKS81}, where $ s(n) \in \Omega(\log(n)) $ is space-constructible.

We begin with a definition:
The prover-verifier pair $ (P,V) $ is an \textit{unbounded-error IPS having perfect completeness} for $ L $
if it has a perfect-completeness, i.e. satisfying condition ($ 1' $), and 
has the following soundness condition:
\begin{enumerate}
	\item[2$ '' $.] for all $ x \notin L $, and all provers $ P^* $, the probability that $ (P^*,V) $ accepts $ x $
		is less than 1.
\end{enumerate}
The classes defined by unbounded-error IPS having perfect completeness will be shown by $ \perunIP{\cdot} $.

\begin{lemma}
	\label{lem:app:proof-of-char-of-perfectIP}
	For any space-constructible $ s(n) \in \Omega(\log(n)) $,
	\[
		\perunIP{s(n)} \subseteq \DSPACE{2^{2^{2^{O(s(n))}}}}.
	\]
\end{lemma}
\begin{proof}
	Let $ \mathtt{L} \in \perunIP{s(n)} $.
	Then there exists an unbounded-error IPS having perfect-completeness $ (P,V) $ for $ \mathtt{L} $.
	We will show that a DTM can recognize $ \mathtt{L} $
	by using triple-exponential time in $ s(n) $.
	
	Let $ x $ be an input string.
	Without loss of generality, we assume that
	the communication alphabet contains exactly two symbols $ \{0,1\} $.
	We represent the configuration set of $ V $ on $ x $ as  $ \mathcal{C}_{V}(x) $,
	whose size is exponential in $ s(|x|) $.
	We classify the configurations into five groups:
	\begin{enumerate}
		\item \texttt{read}: the ones in a reading state
		\item \texttt{comm-0}: the ones in a communication state ready to write 0 on the communication cell
		\item \texttt{comm-1}: the ones in a communication state ready to write 1 on the communication cell
		\item \texttt{acc}: the ones in an accepting state
		\item \texttt{rej}: the ones in a rejecting state
	\end{enumerate}	
	
	The computation tree of $ (P',V) $ on $ x $ can be infinite for a prover $ P' $.
	On the other hand, \textit{having perfect completeness} allow us to build a finite computation tree
	that concisely represents the computation of $ V $ on $ x $ and 
	its communications with all possible provers ($ P^{*} $).
	First of all, we do not need to keep the probabilities of the configurations
	since the input is either accepted with probability 1 or with a probability less than 1.
	(We do not need the cumulative sums of the accepting and rejecting probabilities.)
	Secondly, the length of any halting path must be bounded by a certain number steps.
	If the computation does not halt, then $ V $ must enter an infinite loop. 
	An infinite loop, \textit{in our case}, can either contributes to a halting path 
	or independent from the other parts of the computation. 
	We visualize both cases in Figure \ref{fig:looping}.
	The former case can be seen as a part of the halting path 
	since the probability of being in the loop approaches to zero and the halting path is re-traversed with 
	the decreasing probability after each cycle.
	However, the probability of being in the loop remains the same in latter case,
	and so it should be replaced with a rejecting path if detected.	
	
	\begin{figure}[!ht]
	\centering
	\includegraphics[height=30mm]{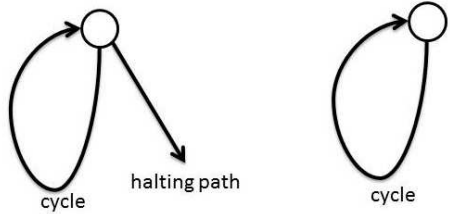}
	\caption{Two cases of infinite loops}
	\label{fig:looping}
	\end{figure}
	
	We denote our finite tree $ \mathcal{T}_{V}(x) $.
	(Note that we do not specify any prover since this tree represent all possible communication scenarios.)
	The structure and evaluation of $ \mathcal{T}_{V}(x) $ is similar to the computation tree of an ATM.
	The main difference of $ \mathcal{T}_{V}(x) $ is that 
	a node can take three different values instead of two values, 
	which are originally \textsf{true} and \textsf{false}.
	We give the details of how $ \mathcal{T}_{V}(x) $ can be constructed below.
	
	Since the protocol is private, we keep the configurations that follow the same communication strategy together.
	Therefore, each node of $ \mathcal{T}_{V}(x) $ represents a subset of $ \mathcal{C}_{V}(x) $.
	The root represents the initial configuration.
	We have four different types of inner nodes, i.e.
	\begin{center}
		\textsc{read-comm}, \textsc{comm-01}, \textsc{comm-0}, and \textsc{comm-1},
	\end{center}
	and
	three different types of leafs, i.e.
	\begin{center}
		\textsc{acc}, \textsc{rej}, and \textsc{loop}.
	\end{center}	
	A \textsc{read-comm} node, say \textsc{RC}, is a node that contains 
	at least one \texttt{read} configuration and may contain 
	some \texttt{comm-0} or \texttt{comm-1} configurations.
	The child node(s) of \textsc{RC} is (are) determined as follows:
	Each \texttt{read} configuration in \textsc{RC} divides into
	at most two child configurations in a single step.
	\begin{itemize}
		\item If the child is an \texttt{acc} (a \texttt{rej}) configuration, 
			then an \textsc{acc} (a \textsc{rej}) leaf is created and connected to the \textsc{RC}.
		\item If the child is a \texttt{read} configuration which is identical to 
			a \texttt{read} configuration in \textsc{RC},
			then a \textsc{loop} leaf is created and connected to \textsc{RC}.
		\item In any other case, each child should be one of 
			a \texttt{read} (not in \textsc{RC}), 
			a \texttt{comm-0}, or a \texttt{comm-1} configuration.
			All these children with the previous \texttt{comm-0} or \texttt{comm-1}
			configurations
			form a new node and connected to \textsc{RC}.
	\end{itemize}	
	The type of the new node given in the last item is determined conditionally.
	If its depth (in $ \mathcal{T}_{V}(x) $) 
	exceeds a certain number ($ 2^{| \mathcal{C}_V(x) |} $ -- the total number of all subsets of $ \mathcal{C}_V(x) $),
	it becomes a \textsc{loop} leaf since it must be a repetition of a previous
	node along the same path.
	
	Suppose that the depth of the new node does not exceed this number.
	If it contains at least one \texttt{read} configuration, 
	then it becomes a \textsc{read-comm} node again.
	If it contains both \texttt{comm-0} and \texttt{comm-1} configurations,
	then it becomes a \textsc{comm-01} node.
	If it contains only some  \texttt{comm-0} (\texttt{comm-1}) configurations,
	then it becomes a \textsc{comm-0} (\textsc{comm-1}) node.

	For each \textsc{comm-01} node, say \textsc{C01}, two new nodes are created and connect to \textsc{C01}.
	One of them becomes a \textsc{comm-0} node that contains all \texttt{comm-0} configurations of \textsc{C01}.
	The other one becomes a \textsc{comm-1} node that contains all \texttt{comm-1} configurations of \textsc{C01}.
	Both \textsc{comm-0} and \textsc{comm-1} are the communication nodes.
	We give the details for a \textsc{comm-0}, say \textsc{C0}, node.
	(The situation is exactly the same for a \textsc{comm-1} node.)
	Let $ c $ be a configuration in \textsc{C0}.
	In $ c $, the verifier writes $ 0 $ on the communication cell, and then
	receives $ 0 $ or $ 1 $. 
	Since we consider all possible communications,
	there are two next configurations evolved from $ c $ in a single step, 
	say $ c_0 $ and $ c_1 $.
	If  $ c_0 $ ($ c_1 $) is an \texttt{acc} or a \texttt{rej} configuration,
	then an \textsc{acc} or a \textsc{ref} leaf is created, respectively, and connected to the \textsc{C0}; or
	if $ c_0 $ ($ c_1 $) is a \texttt{comm-0} configuration which is identical to 
	a \texttt{comm-0} configuration in \textsc{C0},
	then a \textsc{loop} leaf is created and connected to \textsc{C0}.
	Otherwise, all $ c_0 $'s and $ c_1 $'s are collected into two new nodes and then
	connected to \textsc{C0}. The types of the new nodes are determined as explained above.
	We completed how $ \mathcal{T}_{V}(x) $ can be constructed.
	
	Now, we describe how $ \mathcal{T}_{V}(x) $ can be evaluated.
	We associate each inner node with ``$ \wedge $" or ``$ \vee $" operator:
	\textsc{read-comm} and \textsc{comm-01} are associated with ``$ \wedge $" (universal choice),
	and \textsc{comm-0} and \textsc{comm-1} are associated with ``$ \vee $" (nondeterministic choice).
	These operators determines the value of a node from the values of its children.
	There are three types of values: \textsf{true}, \textsf{false}, and \textsf{loop[depth]},
	where \textsf{depth} is a numeric value.
	Obviously, any \textsc{acc} (\textsc{rej}) leaf takes the value of \textsf{true} (\textsf{false}).
	Any \textsc{loop} leaf takes the value of \textsf{loop} with a \textsf{depth} value
	that represent the smallest depth of the repeated node.	
	
	If the computation tree just contains \textsf{true} and \textsf{false} values,
	then its evaluation becomes trivial (exactly the same as alternation).
	The non-trivial part is how to include the \textsf{loop} into the evaluation.
	Suppose that a  node associated with ``$ \wedge $" has $ k>0 $ child/children, whose values are represented
	by $ v_1, \ldots, v_k $.
	The value of the node can be calculated as follows:
	\[
		v_1 \wedge ( v_2 \wedge \cdots (v_{k-2} \wedge ( v_{k-1} \wedge v_k ) ) ),
	\]
	where binary relation ``$ \wedge $" is defined as follows:
	\[
		\begin{array}{|l|c|c|c|}
			\hline
			\multicolumn{1}{|c|}{ \wedge } & \mathsf{~~~~true~~~~} & \mathsf{~~~~false~~~~} & \mathsf{loop[d_1]}
			\\
			\hline
			\mathsf{true} & \mathsf{true} & \mathsf{false} & \mathsf{true}
			\\
			\hline
			\mathsf{false} & \mathsf{false} & \mathsf{false} & \mathsf{false}
			\\
			\hline
			\mathsf{loop[d_2]} & \mathsf{true} & \mathsf{false} & \mathsf{loop[\mathit{min}\{d_1,d_2\}]}
			\\
			\hline
		\end{array}
	\] 
	Since it is a universal choice, any \textsf{false} value beats all the other values.
	Any \textsf{true} value beats any \textsf{loop} value, since this loop contributes the accepting path.
	Between two \textsf{loop} values, we select the one having smaller depth.
	The value of a ``$ \vee $" node can be calculated in a similar way with its specific rules:
	\[
		\begin{array}{|l|c|c|c|}
			\hline
			\multicolumn{1}{|c|}{ \vee } & \mathsf{~~~~true~~~~} & \mathsf{~~~~false~~~~} & \mathsf{loop[d_1]}
			\\
			\hline
			\mathsf{true} & \mathsf{true} & \mathsf{true} & \mathsf{true}
			\\
			\hline
			\mathsf{false} & \mathsf{true} & \mathsf{false} & \mathsf{loop[d_1]}
			\\
			\hline
			\mathsf{loop[d_2]} & \mathsf{true} & \mathsf{loop[d_2]} & \mathsf{loop[\mathit{min}\{d_1,d_2\}]}
			\\
			\hline
		\end{array}
	\] 
	Since it is an existential (nondeterministic) choice, any \textsf{true} value beats all the other values.
	Any \textsf{loop} value beats any \textsf{false} value,
	since the corresponding loop may contribute an accepting path in a lower depth.
	Between two \textsf{loop} values, we again select the one having smaller depth.
	The last thing about the evaluation is that 
	if a node takes value of \textsf{loop} and the depth of the loop refers to this node,
	then the value of the node is changed to \textsf{false}
	since this loop does not contribute any accepting path.
	The value of the root is the value of the tree.
	
	In case of $ x \in \mathtt{L} $, we know that there exists a nondeterministic strategy 
	(corresponding to the communication with $ P $) on the tree
	such that it does not lead to any rejecting path, and any infinite loop must contribute some accepting paths.
	Therefore, the value of the root is set to \textsf{true}.
	
	In case of $ x \notin \mathtt{L} $, we know that for every nondeterministic strategy
	(corresponding to the communication with $ P^{*} $),
	there must be a nonzero rejecting path or a nonzero looping path (not contributing any accepting path)
	that dominates all the other opponent values during the evaluation.
	Therefore, the value of the root is set to \textsf{false}.
	
	A DTM can construct and evaluate $ \mathcal{T}_V(x) $ in a straightforward way.
	Since the depth of the tree is $ 2^{| \mathcal{C}_V(x) |} $,
	the total running time of the DTM is double-exponential in $ |\mathcal{C}_V(x)| $.
	Since $ |\mathcal{C}_V(x)| $ is exponential in $ s(|x|) $,
	then the total running time becomes triple-exponential in $ s(|x|) $.
\end{proof}

\section{A constant-space qAM protocol for $ \subsetsum$}
\label{app:qAM-for-subsetsum}

In this appendix, we present a qAM system having a finite-state verifier for the well-known NP-complete language 
$ \subsetsum $, which is the collection of all strings of the form 
$ S \dollar a_{1} \dollar  \ldots \dollar a_{n} \dollar $
such that $ S $ and the $ a_{i} $'s are numbers in binary ($ 1 \le i \le n $),
and there exists a set $ I \subseteq \{1,\ldots,n\} $ satisfying $ \sum_{i \in I} a_{i} = S $, where $ n>0 $.

\begin{lemma}
	\label{lem:subset-sum}
	$ \subsetsum \in \qAM{1} $.
\end{lemma}
\begin{proof}
	We assume that the input to be of the form  
	$ S \dollar a_{1} \dollar  \ldots \dollar a_{n} \dollar, $
	where $ S $, the $ a_{i} $'s are numbers in binary ($ 1 \le i \le n $), and $ n>0 $.
	(If not, the input is immediately rejected.)
	The input is written between two $ \# $ symbols on the input tape 
	and its head is not allowed to cross these boundaries. 
	
	The main idea is that the verifier scans the input from left to right in an infinite loop
	and firstly encodes $ S $, and then subtracts the encoding of each of the  $ a_i $'s selected by the prover, 
	in some amplitudes of the states on the quantum register.
	And at the end of the loop (round), the verifier tests whether the result is zero or not (described later).
	Since our encoding procedure works by reducing the amplitude with a constant in each step,
	the process can successfully be ended 
	with an exponentially small probability depending on the length of the input.
	Therefore, a new round is initiated with remaining huge probability.
	
	The register has 3 states, $ \{ q_{1}, q_2, q_{3} \} $.
	Each round is composed by five parts, described below.
	The details of superoperators applied in each round
	are given in Figure \ref{fig:superoperator-for-subsetsum},
	where each outcome is given under the operation element.
	The actions associated to each outcome are as follows:
	(i) The input head is moved forward if outcome ``$ f $" is observed,
	(ii) the input is accepted (rejected) if outcome ``$ a $" (``$ r $") is observed, and
	(iii) a new round is initiated if outcome ``$ i $" is observed.
	\begin{enumerate} 
		\item The finite register is initialized on symbol $ \# $:
			$ \ket{\psi_{0}} = (1 ~~ 0 ~~ 0)^{T} $.
		\item $ S $ is encoded into the amplitudes of $ \ket{q_{2}} $:
			$ \mathcal{E}_{\sigma} $ is applied on the quantum register when reading $ \sigma \in \{0,1\} $.
			Then, $ \mathcal{E}_{\dollar} $ is applied on the quantum register when reading $ \dollar $.
		\item Each $ a_{i} $  ($ 1 \leq i \leq n $) is encoded into the amplitude of $ \ket{q_{3}} $:
			$ \mathcal{E}_{\sigma}' $ is applied on the quantum register when reading $ \sigma \in \{0,1\} $.
		\item If an $ a_{i} $ ($ 1 \leq i \leq n $) is \textit{selected} by the prover on symbol $ \dollar $, 
			it is subtracted from the number represented by the amplitude of $ \ket{q_{2}} $:
			$ \mathcal{E'}_{\dollar} $ is applied on the quantum register.
			If it is not selected, $ \mathcal{E}_{\dollar}'' $
			is applied on the quantum register.
			Note that, the amplitude of $ \ket{q_{3}} $ is set to 0 after each of these transformations.
		\item The decision is given on $ \# $:
		$ \mathcal{E}_{\#} $ is applied on the quantum register when reading $ \# $.
	\end{enumerate}
	\begin{figure}[!ht]
		\centering
		\footnotesize
		\begin{tabular}{|ll|l|}
			\hline
			Part 2: 
			& & ~
				$ \mathcal{E}_{0} \mspace{-6mu} = \mspace{-6mu} $
				{\tiny
				$ 
					\left\lbrace
						\underbrace{
						\mspace{-2mu}
						\frac{1}{3}
						\mspace{-2mu}
						\left( 
						\mspace{-10mu}
						\begin{array}{rrr}
							1 & 0 & 0 \\
							0 & 2 & 0 \\
							0 & 0 & 1  \\
						\end{array}
						\mspace{-10mu}
						\right)
						}_{f},
						\underbrace{
						\frac{1}{3}\left( 
						\mspace{-10mu}
						\begin{array}{rrr}				
							2 & 0 & \mspace{-15mu} -2  \\
							2 & 0 & 2 \\
							0 & 2 & 0 \\
						\end{array}
						\mspace{-10mu}
						\right)
						}_{i},
						\underbrace{
						\frac{1}{3}\left( 
						\mspace{-10mu}
						\begin{array}{rrr}	
							0 & 1 & 0 \\
							0 & 0 & 0 \\
							0 & 0 & 0 \\
						\end{array}
						\mspace{-10mu}
					\right)
					}_{i}
					\right\rbrace ,
				$
				}
				$ \mathcal{E}_{1} \mspace{-6mu} = \mspace{-10mu}  $
				{ \tiny
				$ 					
					\left\lbrace
						\underbrace{
						\frac{1}{3}\left( 
						\mspace{-10mu}
						\begin{array}{rrr}
							1 & 0 & 0 \\
							1 & 2 & 0  \\
							0 & 0 & 1  \\
						\end{array}
						\mspace{-10mu}
						\right)
						}_{f},
						\underbrace{
						\frac{1}{3}\left( 
						\mspace{-10mu}
						\begin{array}{rrr}				
							2 & \mspace{-15mu} -1 & 0 \\
							1 & 0 & 2 \\
							1 & 0 & \mspace{-15mu} -2 \\
						\end{array}
						\mspace{-10mu}
						\right)
						}_{i},
						\underbrace{
						\frac{1}{3}\left( 
						\mspace{-10mu}
						\begin{array}{rrr}	
							1 & 0 & 0 \\
							0 & 2 & 0 \\
							0 & 0 & 0 \\
						\end{array}
						\mspace{-10mu}
					\right)
					}_{i}
					\right\rbrace
				$
				}
			\\
			\hline
			Part 2: & & ~
			$ \mathcal{E}_{\mspace{-3mu}\dollar} \mspace{-6mu} = \mspace{-10mu} $
			{ \tiny
				$
				\left\lbrace
					\underbrace{\frac{1}{3} \mathcal{I}}_{f}, 
					\underbrace{\frac{1}{3} 2\mathcal{I}}_{i}, 
					\underbrace{\frac{1}{3} 2\mathcal{I}}_{i}
				\right\rbrace 
				$}
			\\
			\hline
			Part 3: & & ~
				$ \mathcal{E}_{0}' = $
			{\tiny
				$ 
					\left\lbrace
						\underbrace{
						\frac{1}{3}\left( 
						\mspace{-10mu}
						\begin{array}{rrr}
							1 & 0 & 0 \\
							0 & 1 & 0 \\
							0 & 0 & 2 \\
						\end{array}
						\mspace{-10mu}
						\right)
						}_{f},
						\underbrace{
						\frac{1}{3}\left( 
						\mspace{-10mu}
						\begin{array}{rrr}				
							2 & 2 & 0 \\
							2 & \mspace{-10mu} -2 & 0 \\
							0 & 0 & 2 \\
						\end{array}
						\mspace{-10mu}
						\right)
						}_{i},
						\underbrace{
						\frac{1}{3}\left( 
						\mspace{-10mu}
						\begin{array}{rrr}	
							0 & 0 & 1 \\
							0 & 0 & 0 \\
							0 & 0 & 0 \\
						\end{array}
						\mspace{-10mu}
					\right)
					}_{i}
					\right\rbrace,
				$
			}
			$ \mathcal{E}_{1}' = $
			{ \tiny
				$ 
					\left\lbrace
						\underbrace{
						\frac{1}{3}\left( 
						\mspace{-10mu}
						\begin{array}{rrr}
							1 & 0 & 0 \\
							0 & 1 & 0 \\
							1 & 0 & 2 \\
						\end{array}
						\mspace{-10mu}
						\right)
						}_{f},
						\underbrace{
						\frac{1}{3}\left( 
						\mspace{-10mu}
						\begin{array}{rrr}				
							2 & 0 & \mspace{-10mu} -1 \\
							1 & 2 & 0 \\
							1 & \mspace{-10mu} -2 & 0 \\
						\end{array}
						\mspace{-10mu}
						\right)
						}_{i},
						\underbrace{
						\frac{1}{3}\left( 
						\mspace{-10mu}
						\begin{array}{rrr}	
							1 & 0 & 0 \\
							0 & 0 & 2 \\
							0 & 0 & 0 \\
						\end{array}
						\mspace{-10mu}
					\right)
					}_{i}
					\right\rbrace
				$
			}
			\\
			\hline
			Part 4: & & ~
				$ \mathcal{E}_{\dollar}' =  $
				{\tiny
				$ 
					\left\lbrace
					\underbrace{
					\frac{1}{3}\left( 
					\mspace{-10mu}
					\begin{array}{rrr}
						1 & 0 & 0 \\
						0 & 1 & \mspace{-10mu} -1 \\
						0 & 0 & 0 \\
					\end{array}
					\mspace{-10mu}
					\right)
					}_{f},
					\underbrace{
					\frac{1}{3}\left(
					\mspace{-10mu} 
					\begin{array}{rrr}				
						0 & \mspace{-10mu} -1 & 1 \\
						2 & 1 & \mspace{-10mu} -1 \\
						2 & -1 & 1 \\
					\end{array}
					\mspace{-10mu}
					\right)
					}_{i},
					\underbrace{
					\frac{1}{3}\left( 
					\mspace{-10mu}
					\begin{array}{rrr}	
						0 & 2 & 2 \\
						0 & 0 & 0 \\
						0 & 0 & 0 \\
					\end{array}
					\mspace{-10mu}
					\right)
					}_{i},
					\underbrace{
					\frac{1}{3}\left( 
					\mspace{-10mu}
					\begin{array}{rrr}	
						0 & 1 & 0 \\
						0 & 0 & 1 \\
						0 & 0 & 0 \\
					\end{array}
					\mspace{-10mu}
					\right)
					}_{i}			
				\right\rbrace,
				$
			}
			$ \mathcal{E}_{\dollar}'' =  $
			{\tiny
				$
					\left\lbrace
					\underbrace{
					\frac{1}{3}\left( 
					\mspace{-10mu}
					\begin{array}{rrr}
						1 & 0 & 0 \\
						0 & 1 & 0 \\
						0 & 0 & 0 \\
					\end{array}
					\mspace{-10mu}
					\right)
					}_{f},
					\underbrace{
					\frac{1}{3}\left( 
					\begin{array}{rrr}				
						2 & \mspace{-10mu} -2 & 0 \\
						\mspace{-10mu} 2 & 2 & 0 \\
						0 & 0 & 3 \\
					\end{array}
					\mspace{-10mu}
					\right)
					}_{i}			
				\right\rbrace
				$
			}
		\\
		\hline
		Part 5: & & ~
			$ \mathcal{E}_{\#} =  $
		{\tiny
			$
				\left\lbrace
					\underbrace{
					\frac{1}{3}\left( 
					\mspace{-10mu}
					\begin{array}{rrr}
						1 & 0 & 0 \\
						0 & 0 & 0 \\
						0 & 0 & 0 \\
					\end{array}
					\mspace{-10mu}
					\right)
					}_{a},
					\underbrace{
					\frac{1}{3}\left( 
					\mspace{-10mu}
					\begin{array}{rrr}				
						0 & 0 & 0 \\
						0 & 3 & 0 \\
						0 & 0 & 0 \\
					\end{array}
					\mspace{-10mu}
					\right)
					}_{r},
					\underbrace{
					\frac{1}{3}\left( 
					\mspace{-10mu}
					\begin{array}{rrr}	
						2 & 0 & 0 \\
						2 & 0 & 0 \\
						0 & 0 & 3 \\
					\end{array}
					\mspace{-10mu}
				\right)
				}_{i}
				\right\rbrace
			$
		}
		\\ 
		\hline
		\end{tabular}	
		\label{fig:superoperator-for-subsetsum}
		\caption{The details of superoperators used by the qAM system for $ \subsetsum $}
	\end{figure}		
	
	Let $ w $ be the input and $ T $ be the cumulative sum of selected $ a_{i} $'s  by the prover.
	Then, the state of the register before reading $ \# $ becomes
	{\footnotesize
	\[
		\ket{ \widetilde{ \psi_{|w|} } } = \left( \frac{1}{3} \right)^{|w|} 
		\left( \mspace{-10mu} \begin{array}{c}
						1 \\ S-T \\ 0
					\end{array} \mspace{-10mu} \right).
	\]
	}
	After applying $ \mathcal{E}_{\#} $, the input is rejected with probability
	\footnotesize
	\[
		\left( \frac{1}{3} \right)^{2|w|+2} 
		\left( 3 S -3 T \right)^{2},
	\]
	\normalsize
	which is at least $ 9 \left( \frac{1}{3} \right)^{2|w|+2} $ if $ S \neq T $
	and is exactly equal to 0 if $ S=T $.
	On the other hand, the input is always accepted with probability $ \left( \frac{1}{3} \right)^{2|w|+2} $.
	Therefore, if $ w \in \subsetsum  $, there exists a prover such that it is accepted exactly, and
	if $ w \notin \subsetsum  $,
	whatever the prover says, it is rejected with a probability at least $ \frac{9}{10} $.
	The error bound can be reduced to any desired value by using probability amplification.
\end{proof}

\section{A time-bound for absolutely-halting space-bounded quantum Turing machines}
\label{app:time-bound-for-AH-QTMs}

The first space-bounded QTM model was introduced  by Watrous \cite{Wa98,Wat99A}, 
and he showed that any such $ s(n) $ space-bounded QTM
that always halts on every input can run at most $ 2^{O(s(n))} $ steps,
where $ s(n) \in \Omega(\log(n)) $.
Since the nonhalting part of such a model can always be in a single pure (quantum) state,
the same result cannot be directly applicable to the QTMs 
whose halting part can be in a mixture of some pure states (mixed-state).
The qATM introduced in Section \ref{sec:q-alternation}
and the models introduced in \cite{Wa03,YS11A,vMW12} are some examples for the latter case.

On the other hand, since any mixed-state and the operator(s) applied to it can be represented
by a single vector and a single matrix, respectively, the result given by Watrous can be extended to general case.
We will provide an explicit proof of this result below.

A QTM can have both classical and quantum parts.
Let $ \mathcal{M} $ be such a space-bounded QTM and $ x $ be an input.
A standard configuration of $ \mathcal{M} $ on $ x $ is a pair of $ (c,\ket{d}) $,
where $ c $ is a configuration of the classical part and
$ \ket{d} $ is a (standard) basis vector of the quantum part.
During the computation, $ \mathcal{M} $ can be in some mixture of $ (c,\ket{\psi}) $'s,
where each $ \ket{\psi} $ can be either a basis vector or a superposition of some basis vectors.

\begin{theorem}
	\label{thm:time-bound-for-AH-QTMs}
	Let $ N $ be the number of the standard configurations of 
	an absolutely-halting space-bounded QTM $ \mathcal{M} $ on a given input $ x $.
	Then, $ \mathcal{M} $ can run at most $ N^2 $ steps on $ x $.
\end{theorem}
\begin{proof}
	In space-bounded quantum computation, the computation is regularly observed 
	whether it is terminated or continued,
	and then the observed part is normalized.
	The nonhalting part of $ \mathcal{M} $ on $ x $ can be represented by an 
	$ N \times N $-dimensional density matrix.
	Since we consider whether this matrix is equal to zero matrix or not, 
	we can omit the normalization part.
	Based on the local transitions of $ \mathcal{M} $, we can defined some finite, say $ k $, $ N \times N $ matrices
	$ \{ E_1,\ldots,E_k \} $
	that represent one step transformation of the nonhalting part.
	Thus, we obtain the following matrix sequence that represent the nonhalting part for each step:
	\begin{equation}
		\label{eq:nonhalting-part}
		\nu_0,\nu_1,\nu_2,\ldots ~,
	\end{equation}
	where $ \nu_0 $ is the initial one and $ \nu_i $ represents the $ i^{th} $ ($ i>0 $) one obtained after $ i^{th} $  steps,
	which is calculated as:
	\[
		\nu_i = \sum_{j=1}^{k}  E_j \nu_{i-1} E_j^{\dagger}.
	\]
	If $ \mathcal{M} $ halts on every input absolutely, there must be an index $ i' $ such that $ \nu_{i'} = 0 $.
	As pointed out above, the sequence given in Eq. \ref{eq:nonhalting-part} can be represented by vectors, and
	each of them (except the initial one) can be obtained by applying a single operator (matrix) 
	to the previous vector in the sequence.
	That is, based on $ \nu_0 $ and $ \{ E_1,\ldots,E_k \} $, 
	we define an $ N^2 $-dimensional vector, say $ v_0 $,  and  $ N^2 \times N^2 $-dimensional matrix, say $ E $,
	respectively, and then Eq. \ref{eq:nonhalting-part} turns out be as follows:
	\[
		v_0,v_1,v_2,\ldots~,
	\]
	where 
	\[
		v_i = E v_{i-1} ~~ (i > 0).
	\]
	We refer the reader to Page 73 of \cite{Wa03} for the details of this conversion.
	Now, we will show that $ i' $ cannot be bigger than $ N^2 $ due to the following fact.
	
	\begin{fact}
	For any $ m $-dimensional vector $ u $, if $ A^{m}u \neq 0  $, then $ A^{m+j}u \neq 0 $,
	where $ A $ is an $ m \times m $-dimensional matrix and $ j>0 $.\footnote{Observe 
	that one can easily find an example of $ A $ and $ u $ such that $ A^{m-1}u \neq 0 $ but $ A^{m}u = 0 $.}
	\end{fact}
	\begin{proof}
		The proof can be easily obtained from the following well-known relation:
		\[
			ker(A) \subseteq ker(A^2) \subseteq \cdots \subseteq ker(A^m) = ker(A^{m+1}) = ker(A^{m+2}) = \cdots
		\]	
		That is, if $ u $ is not in $ ker(A^m) $, then $ u $ cannot be in $ ker(A^{m+j}) $ for any $ j>0 $.
	\end{proof}
	Thus we can say that if $ E v_{N^2} \neq 0 $, then $ v_{N^2+j} $ cannot be equal to zero for any $ j>1 $,
	i.e. $ \mathcal{M} $ cannot halt absolutely on $ x $.
	Therefore, $ i' $ cannot be bigger than $ N^2 $, 
	which is a quadratic bound in terms of the number of configurations.
\end{proof}

\begin{corollary}
	Any $ s(n) \in \Omega(\log(n)) $ space-bounded QTM that always halt on every input
	can run at most $ 2^{O(s(n))} $ steps.
\end{corollary}

\bibliographystyle{alpha}
\bibliography{YakaryilmazSay}

\end{document}